\documentclass[11pt]{article}  

\usepackage{graphicx}

\usepackage{latexsym}
\usepackage[english]{babel}
\usepackage{amsmath}
\usepackage{amsfonts}
\usepackage{amssymb}
\usepackage{makeidx}
\usepackage[utf8]{inputenc}
\usepackage{lmodern}
\usepackage{verbatim}
\usepackage[pagewise]{lineno}
\usepackage{amsthm}
\usepackage{tikz-cd} 

\usepackage[style=authoryear, backend=bibtex, maxbibnames=99]{biblatex}
\newtheorem{theorem}{Theorem}[section]
\newtheorem{corollary}[theorem]{Corollary}

\newtheorem{proposition}[theorem]{Proposition}

\theoremstyle{definition}
\newtheorem{definition}[theorem]{Definition}
\newtheorem{remark}[theorem]{Remark}
\newtheorem{example}[theorem]{Example}

\newcommand*{\Reeb}{\mathcal{R}}
\newcommand{\R}{\ensuremath{\mathbb{R}}}

\newcommand{\F}{\ensuremath{\mathbb{F}}}
\newcommand*{\contr}[1]{\iota_{#1}}

\DeclareMathOperator{\ad}{ad}
\DeclareMathOperator{\Ad}{Ad}

\DeclareMathOperator{\tr}{tr}

\begin{document}
	\title{Reduction by symmetries of contact mechanical systems on Lie groups}
	\author{
		{\bf\large Alexandre Anahory Simoes$^{1}$, Leonardo Colombo$^{2}$,}
		\\
		{\bf\large Manuel de León$^{3}$}, {\bf\large Juan Carlos Marrero$^{4}$,} \\{\bf\large David Martín de Diego$^{5}$, Edith Padrón$^{6}$}\hspace{2mm}
		\vspace{1mm}\\
		{\it\small $^{1}$ IE School of Science and Technology, Madrid, Spain }\\
        {\it\small $^{2}$ Centre for Automation and Robotics (CSIC-UPM),}
        {\it\small Arganda del Rey, Spain}\\
        {\it\small $^{3}$ Instituto de Ciencias Matematicas (CSIC) and}\\
        {\it\small Real Academia de Ciencias, Madrid, Spain }\\
		{\it\small $^{4, 6}$ ULL-CSIC Geometría Diferencial y Mecánica Geométrica,}\\
        {\it\small Departamento de  Matemáticas, Estadística e Investigación Operativa and }\\
        {\it\small Instituto de Matemáticas y Aplicaciones (IMAULL), Faculty of Sciences,}\\
        {\it\small University of La Laguna, Canary Islands, Spain }\\
        {\it\small $^{5}$ Instituto de Ciencias Matematicas  (CSIC), Madrid, Spain }}

\maketitle

\begin{abstract}
	We study the dynamics of contact mechanical systems on Lie groups that are invariant under a Lie group action. Analogously to standard mechanical systems on Lie groups, existing symmetries allow for reducing the number of equations. Thus, we obtain Euler-Poincaré-Herglotz equations on the extended reduced phase space $\mathfrak{g}\times \R$ associated with the extended phase space $TG\times \R$, where the configuration manifold $G$ is a Lie group and $\mathfrak{g}$ its Lie algebra. Furthermore, we obtain the Hamiltonian counterpart of these equations by studying the underlying Jacobi structure. Finally, we extend the reduction process to the case of symmetry-breaking systems which are invariant under a Lie subgroup of symmetries.
\end{abstract}

\let\thefootnote\relax\footnote{\noindent AMS {\it Mathematics Subject Classification (2020)}. Primary 70G45, 70G65, 70G75; Secondary  53D05, 53D10, 53D17.
	 }

\section{Introduction}

Contact Hamiltonian and Lagrangian systems have deserved a lot of attention in recent years (see \cite{Bravetti2017,Bravetti2018}, \cite{deLeon2018,deLeon2019}, \cite{anahory2021geometry}, and references therein).
One of the most relevant features of contact dynamics is the absence of conservative properties
contrarily to the conservative character of the energy in symplectic dynamics; indeed, 
they have a dissipative behavior. This fact suggests that contact geometry may be the appropriate framework to model many physical and mathematical problems with dissipation we find in thermodynamics, statistical physics, systems with impulsive effect and discontinuities, quantum mechanics, gravity or control theory, among many others (see \cite{de2020review}, \cite{colombo2022contact}, \cite{lopez2022nonsmooth}, \cite{de2023optimal}, and references therein).

As an illustrative example, consider the motion of a rigid body under a linear dissipation term. The equations of motion for this kind of systems are of the general form
\begin{equation}\label{damped:rigid:body}
\mathbb{I}\dot{\xi}=\mathbb{I}\xi \times \xi + C\xi,
\end{equation}
where $\xi\in\mathbb{R}^{3}$ denotes the body angular velocity, $\mathbb{I}\in\mathbb{R}^{3\times 3}$ is the inertia tensor, and $C\in \mathbb{R}^{3 \times 3}$ denotes a constant positive semi-definite matrix called the damping matrix (see \cite{shen2003asymptotic}). In the special case where $C=\gamma \mathbb{I}$, with $\gamma$ a non-zero real number, we will see that equations \eqref{damped:rigid:body} are given by a contact version of Euler-Poincaré equations (see \cite{holm1998euler}) for a reduced contact Lagrangian.




As it is noted in \cite{holm2009geometric}, in the case of the rigid body on $SO(3)$ we can rewrite the Lagrangian function $L:TSO(3)\to\mathbb{R}$ for this system as
\begin{equation*}
L(R,\dot{R})=\frac{1}{2}\tr(\dot{R}\mathbb{I}\dot{R}^{T}),
\end{equation*} with $\mathbb{I}$ symmetric and constant. Moreover, due to the $SO(3)$-invariance of $L$, we may define the reduced Lagrangian $l:\mathfrak{so}(3) \rightarrow \mathbb{R}$ given by
\begin{equation*}
l(\xi)=\frac{1}{2}\langle\mathbb{I}\xi,\xi\rangle.
\end{equation*}

The Euler-Poincaré equations applied to the reduced Lagrangian function $l$ give the Euler equations for the rigid body:
\begin{equation*}
\mathbb{I}\dot{\xi}=\mathbb{I}\xi \times \xi.
\end{equation*}
Now, we can consider a slight modification of these constructions and start with a contact Lagrangian function $L: TSO(3)\times \mathbb{R} \rightarrow \mathbb{R}$ given by

\begin{equation*}
L(R,\dot{R},z)=\frac{1}{2}\tr(\dot{R}\mathbb{I}\dot{R}^{T}) + \gamma z,
\end{equation*} where $\gamma\in\mathbb{R}$. This motivates the definition of a contact reduced Lagrangian function $l:\mathfrak{so}(3)\times \mathbb{R} \rightarrow \mathbb{R}$ given by
\begin{equation}\label{rigid:body:contact}
l(\xi, z)=\frac{1}{2}\langle\mathbb{I}\xi,\xi\rangle + \gamma z.
\end{equation}
The main purpose of this work is to introduce the \textit{Euler-Poincaré-Herglotz equations} as a generalization of Euler-Poincaré equations for contact Lagrangian functions, which are given by the expression
\begin{equation*}
\frac{d}{dt}\frac{\delta l}{\delta \xi}=\text{ad}_{\xi}^{*}\frac{\delta l}{\delta \xi}+\frac{\delta l}{\delta \xi}\frac{\partial l}{\partial z}.
\end{equation*}
Applying them to the Lagrangian function \eqref{rigid:body:contact}, we obtain equations \eqref{damped:rigid:body} with $C=\gamma \mathbb{I}$.

In this situation, the reduced Legendre transform $\mathbb{F}l:\mathfrak{g}\times\mathbb{R}\to\mathfrak{g}^{*}\times\mathbb{R}$ is given by the map $\F l(\xi,z)=(\mathbb{I}\xi,z)$. Assuming that $\mathbb{I}$ is a positive definite symmetric matrix, so that $\F l$ is a diffeomorphism and we might define the Hamiltonian function
	$$h(\mu,z)= \frac{1}{2}\mu^{T}\mathbb{J}\mu + \gamma z,$$ where $\mathbb{J}=\mathbb{I}^{-1}$. 
	Given a regular reduced Lagrangian $l$, if $(\xi(t),z(t))$ is a solution of Euler-Poincaré-Herglotz equations and $\mu(t)=\frac{\delta l}{\delta \xi}(\xi(t),z(t))$ we will see in this paper that the curve $t \mapsto (\mu(t),z(t))$ satisfies
	\begin{equation}\label{liepoisson}
			\dot{\mu} =\text{ad}^{*}_{\frac{\delta h}{\delta \mu}}\mu-\mu\frac{\partial h}{\partial z}, \,\,
			\dot{z}  = \langle \mu,\frac{\delta h}{\delta \mu} \rangle-h(\mu,z).
	\end{equation}

In addition, we will define a Lie-Poisson-Jacobi structure and bracket  on $\mathfrak{g}^{*}\times \R$, and derive the corresponding Lie-Poisson-Jacobi equations \eqref{liepoisson}. The Jacobi structure on the vector space $\mathfrak{g}^{*}\times \R$ is linear. This kind of structures was discussed in \cite{iglesias2000some, iglesias2001some}.

Symmetry breaking is common in several physical contexts, from classical mechanics to particle physics and complex fluids (see \cite{marsden1984semidirect}, \cite{holm1998euler}). The simplest example is the heavy top dynamics (the motion of a rigid body with a fixed point in a gravitational field), where due to the presence of gravity, we get a Lagrangian which is $\mathrm{SO}(2)$-invariant but not $\mathrm{SO}(3)$-invariant, contrary to what happens for the free rigid body. Reduction theory for symmetry breaking with applications to nematic systems was studied in \cite{gay2010reduction}. Optimal Control for systems with symmetry breaking has been studied in \cite{gay2011clebsch}. In the context of motion planning, the symmetry breaking appears naturally in the form of a navigation function \cite{bloch2017optimal}, \cite{leo2023}. More recently it has been employed to reduce necessary conditions for optimality in a collision and obstacle avoidance problem (see \cite{stratoglou2022optimal}, \cite{stratoglou2022reduction}).  Motivated by the rigid body attitude dynamics studied in \cite{shen2003asymptotic}, \cite{cho2003mathematical} for the triaxial attitude control testbed given in \cite{bernstein2001development}, we will study Euler-Poincare-Herglotz reduction for systems with symmetry breaking and we will also derive the corresponding Lie-Poisson-Jacobi equations.

This paper is structured as follows: in Section 2, we briefly recall Euler-Poincaré reduction on Lie groups and in Section 3, we recall the basic concepts involving contact structures and contact Hamiltonian vector fields. In section 4, we present the Euler-Poincaré-Herglotz reduction of invariant contact systems on Lie groups. In section 5, we define the Jacobi bracket on the extended dual of the Lie algebra $\mathfrak{g}^{*}\times \R$ and obtain the Hamiltonian counterpart of Euler-Poincaré-Herglotz equations. Finally, in Section 6, we apply the previous development to systems with a broken symmetry, through a Euler-Poincaré-Herglotz reduction on a semi-direct product.
 
\subsection*{Notation}
In order to be self-contained and use consistent notation, let us establish the functional derivative notation that we will use without reference all along the text.

Let $f$ be a smooth function on a finite-dimensional vector space $V$. The derivative of $f$ at a point $v\in V$ along the direction $u \in V$, may written in terms of geometric objects as $$Df(v)\cdot u = \left. \frac{d}{dt} \right|_{t=0}f(v+tu) = \langle df(v),(u)_{v}^{V} \rangle,$$
where $df(v)\in T_{v}^{*}V$ is the differential of the function at $V$ and $$(u)_{v}^{V}=\left. \frac{d}{dt} \right|_{t=0}(v+tu)\in T_vV$$ is the vertical lift of the vector $u$ at $v$.

The vertical lift at $v$ induces a canonical identification between the vector space $V$ and its tangent space at $v$, i.e., the map $(\cdot)_{v}^{V}:V\rightarrow T_{v}V$ is an isomorphism. This implies that the dual map $\varphi:T_{v}^{*}V\rightarrow V^{*}$ is also an isomorphism, so that $$Df(v)\cdot u = \langle \varphi(df(v)),u \rangle.$$
The unique vector $\varphi(df(v))\in V^{*}$ is called the functional derivative of $f$ at $v$. The \textit{functional derivative} of $f$ at $v$ is usually denoted by
\begin{equation*}
	\varphi(df(v))=\frac{\delta f}{\delta v}(v).
\end{equation*}
In finite dimensions, through a coordinate analysis, we see that the coordinate expression of $\displaystyle{\frac{\delta f}{\delta v}(v)}$ matches the component functions of the differential $df(v)$.

\section{Euler-Poincaré equations}

Let $G$ be a Lie group and let $\mathcal{L}_{g}:G\rightarrow G$ be the left multiplication action. Since $\mathcal{L}_{g}$ is a diffeomorphism, it naturally induces a ``trivialized chart" of $TG$, which more precisely is an identification between each tangent space $T_{g}G$ with the Lie algebra $\mathfrak{g}$ through the following map
\begin{equation}\label{LTC}
	TG \rightarrow G\times \mathfrak{g}, \quad (g,\dot{g})\mapsto (g,T_{g}\mathcal{L}_{g^{-1}}(\dot{g})).
\end{equation}
Consider the Lagrangian function $L$ as a real function $L:G\times\mathfrak{g}\to\mathbb{R}$ on $G\times\mathfrak{g}$. Then, we can express the Euler-Lagrange equations in terms of this trivialization as
\begin{equation*}
	\begin{split}
		\frac{d}{dt}\frac{\delta L}{\delta \xi} & =\text{ad}_{\xi}^{*}\frac{\delta L}{\delta \xi}+T_{e}^{*}\mathcal{L}_{g}\left( \frac{\delta L}{\delta g} \right) \\
		\dot{g} & = T_{e}\mathcal{L}_{g}(\xi).
	\end{split}
\end{equation*}  Here $\ad^{*}_{\xi}:\mathfrak{g}^{*}\rightarrow \mathfrak{g}^{*}$ is the dual of the adjoint operator $\ad_{\xi}:\mathfrak{g}\rightarrow \mathfrak{g}$ defined by $\ad_{\xi}(\eta) = [\xi, \eta]$, for $\eta\in \mathfrak{g}$.

Now, if $L: TG\rightarrow \R$ is a \textit{left-invariant} Lagrangian function, that is, its composition with the tangent lift of left-translations  $T\mathcal{L}_{g}$ leaves $L$ invariant
\begin{equation*}
	L (T_{h}\mathcal{L}_{g}(v_{h}))=L(v_{h}), \ \forall v_{h}\in T_{h} G,
\end{equation*}
then the restriction of $L$ to the tangent space at the identity $e\in G$, which we identify with the Lie algebra $\mathfrak{g}$, is called the \textit{reduced Lagrangian function} and we will denote it by $l:\mathfrak{g}\rightarrow \R$. Then all the information is contained in the reduced Lagrangian $l$ since
\begin{equation*}
	l(\eta)=L(T_{e}\mathcal{L}_{g}(\eta)), \quad \forall g \in G, \eta \in \mathfrak{g}.
\end{equation*}
It is well-known that a curve $g:I\rightarrow G$ is a solution of Euler-Lagrange equations if and only if the curve $\xi:I\rightarrow \mathfrak{g}$ determined by
\begin{equation*}
	\dot{g}=T_{e}\mathcal{L}_{g}(\dot{\xi})
\end{equation*}
satisfies the \textit{Euler-Poincaré} equations given by
\begin{equation*}
	\frac{d}{dt}\frac{\delta l}{\delta \xi}=\text{ad}_{\xi}^{*}\frac{\delta l}{\delta \xi}.
\end{equation*}
If $(y^{i})$ are local coordinates associated to a given basis $\{e_{i}\}$ of $\mathfrak{g}$, then its local expression reads
\begin{equation*}
\frac{d}{dt}\frac{\partial l}{\partial y^{i}}=C_{j i}^{k}y^{j}\frac{\partial l}{\partial y^{k}},
\end{equation*}
where $C_{i j}^{k}$ are the structure constants of the Lie algebra given by
\begin{equation*}
	[e_{i},e_{j}]=C_{i j}^{k} e_{k}
\end{equation*} (see \cite{Bloch, marsden2013introduction} for instance).

\section{Contact manifolds and contact dynamics}

In this section, we will recall the main definition of a \textbf{contact manifold} and the corresponding Hamiltonian vector fields (see~\cite{Arnold1978, marle, deLeon2018} for a more detailed overview).

A contact manifold $(M,\eta)$ is an $(2n+1)$-dimensional manifold equipped with a contact form $\eta$, i.e., $\eta$ is a $1$-form on $M$ such that $\eta \wedge (d \eta)^n$ is a volume form. The Reeb vector field $\Reeb\in \mathfrak{X}(M)$ is the unique vector field that satisfies:
\begin{equation}
	\contr{\Reeb} d \eta = 0, \quad \eta(\Reeb)=1.
\end{equation}

On a contact manifold $(M,\eta)$, we  define the following isomorphism of vector bundles:
\begin{equation}
	\begin{aligned}
		\flat: \quad&TM &\longrightarrow& T^* M,\\
	  	        & \quad v &\longmapsto& \contr{v}d \eta + \eta(v)\eta.
	\end{aligned}
\end{equation}
Notice that $\flat(\Reeb) = \eta$.

There is a Darboux theorem for contact manifolds. In a neighborhood of each point in $M$ one can find local coordinates $(q^i, p_i, z)$ such that
\begin{equation}
	 \eta = d z - p_i  d q^i.
\end{equation}
In these coordinates, we have
\begin{equation}
	\Reeb = \frac{\partial}{\partial z}.
\end{equation}

The canonical example of contact manifold is $T^*Q\times \R$. Here, the contact form is given by 
\begin{equation}\label{eq:cotangent_contact_structure}
	\eta_Q = d z - \theta_Q = d z - p_i d q^i,
\end{equation}
where $\theta_Q$ is the pullback of the tautological $1$-form of $T^*Q$, $(q^i,p_i)$ are bundle coordinates on $T^*Q$ and $z$ is the $\R$-coordinate.

Given a smooth function $f:M\rightarrow\R$, its Hamiltonian vector field $X_f$ is given by
\begin{equation}
	\flat(X_f) =  d f - (f + \Reeb(f)) \eta,
\end{equation} or, equivalently, $$i_{X_f}d\eta=df-R(f)\eta,\quad i_{X_f}\eta=-f.$$

We call the triple $(M, \eta, H)$ a contact Hamiltonian system, where $(M,\eta)$ is a contact manifold and $H:M \rightarrow \R$ is the Hamiltonian function. In contrast to their symplectic counterpart, contact Hamiltonian vector fields do not preserve the Hamiltonian fuction. In fact,
\begin{equation}
	X_H(H) = -\Reeb(H) H.
\end{equation}

\subsection{Contact Lagrangian systems}\label{contact:lagrangian}

Now we recall the Lagrangian picture of contact systems (see ~\cite{deLeon2018,deLeon2019} for a more comprehensive description). Indeed, under a suitable choice of contact structure on the extended phase space $TQ\times\R$, the contact Hamiltonian equations with respect to a Lagrangian function on this space are precisely Euler-Poincaré equations.

Let $Q$ be an $n$-dimensional configuration manifold and consider the extended phase space $TQ \times \R$ and we call {contact Lagrangian function}  to any smooth function of the type 
 $L:TQ\times \R \to \R$. In this paper, we will assume that the Lagrangian is regular, that is, the Hessian matrix with respect to the velocities $(W_{ij})$ is regular where
 \begin{equation}\label{eq:hessian}
     W_{ij} = \frac{\partial^2 L}{\partial \dot{q}^i \partial \dot{q}^j },
 \end{equation}
 and $(q^i, \dot{q}^i,z)$ are bundle coordinates for $TQ \times \R$. Equivalently, $L$ is regular if and only if the one-form
\begin{equation}
    \eta_L = d z - \theta_L
\end{equation}
is a contact form. Here,
\begin{equation*}
    \theta_L = S^* (d L) = \frac{\partial L}{ \partial \dot q^i} d {q}^i,
\end{equation*}
 where $S$ is the canonical vertical endomorphism $S:TTQ \to TTQ$ extended to $TQ \times \R$, that is, in local $TQ \times \R$ bundle coordinates,
\begin{equation}\label{eq:canonical_endomorphism}
    S = d q^i \otimes \frac{\partial}{\partial \dot{q}^i}.
\end{equation}

The energy of the system is defined by
\begin{equation}
    E_L = \Delta(L) - L = \dot{q}^i \frac{\partial L}{\partial \dot{q}^i} - L,
\end{equation}
where $\Delta$ is the Liouville vector field on $TQ$ extended to $TQ\times \R$ in the natural way.

The Reeb vector field of $\eta_L$, which we will denoted by $\Reeb_L$, is locally given by
\begin{equation}\label{eq:flat_iso}
    \Reeb_L = \frac{\partial}{\partial z} - 
    (W^{ij}) \frac{\partial^2 L}{\partial \dot{q}^i \partial z} \frac{\partial}{\partial \dot{q}^j},
\end{equation}
where $(W^{ij})$ is the inverse of the Hessian matrix with respect to the velocities $(W_{ij})$.

The Hamiltonian vector field of the energy $E_L$ will be denoted $\xi_L = X_{E_L}$, which is a second order differential equation (SODE) (that is, $S(\xi_L) = \Delta$) and its solutions are just the ones of the Herglotz equations for $L$ (see~\cite{deLeon2018,deLeon2019}):
\begin{equation}\label{Herglotz:standard}
			\begin{split}
				\frac{d}{dt}\frac{\partial L}{\partial \dot{q}^{i}} & =\frac{\partial L}{\partial q^{i}}+\frac{\partial L}{\partial \dot{q}^{i}}\frac{\partial L}{\partial z} \\
				\dot{z} & = L.
			\end{split}
		\end{equation}

Now, we will recall the definition of the \emph{Legendre transformation} for contact Lagrangian systems. Given the vector bundle $TQ\times \R \to Q \times \R$, one can consider the fiber derivative $\F L$ of $L:TQ \times \R \to \R$, which has the following coordinate expression in natural coordinates:
\begin{equation}\label{Legendre:transform}
	\begin{aligned}
		\F L:TQ \times \R &\to T^*Q \times \R\\
		(q^i,\dot{q}^i,z) &\mapsto (q^i, \frac{\partial L}{\partial \dot{q}^{i}},z).
	\end{aligned}
\end{equation}
If the Lagrangian is regular, then $\F L$ is a local diffeomorphism and one can show that $\eta_L$ is simply $(\F L)^{*}\eta_{Q}$. 

In the case that $\F L$ is a global contactomorphism, then we say that $L$ is \emph{hyperregular}. In this situation, we can define a Hamiltonian function $H:T^*Q\times \R \to \R$ such that $E_L = H \circ \F L$. So,  the Lagrangian and Hamiltonian dynamics are $\F L$-related, that is, $(\F L)_{*} (\xi_L) = X_{H}\circ\mathbb{F}L$ (see \cite{deLeon2018,deLeon2019}).

\section{Euler-Poincaré-Herglotz equations for contact reduced Lagrangian systems}

In this section, we will prove that a similar result to that in Section 2 holds in the contact situation where we consider the extended space $TG\times \R$ as our initial ambient phase space and we assume that the Lagrangian function is $G$-invariant. Let $L:TG\times \R \rightarrow \R$ be a contact Lagrangian function. Again, the fact that the left translation $\mathcal{L}_{g}:G\to G$ is a diffeomorphism allows us to consider a ``left trivialized chart" of $TG\times \R$, that is an identification of the product $T_{g}G\times \R$ with the vector space $\mathfrak{g}\times \R$ given precisely by the map
\begin{equation}\label{LTC:extended}
TG\times \R \rightarrow G\times \mathfrak{g}\times\mathbb{R}, \quad (g,\dot{g},z)\mapsto (g,T_{g}\mathcal{L}_{g^{-1}}(\dot{g}),z).
\end{equation}
Thus, if we introduce coordinates $(g,\xi,z)$ on $TG\times \R$ given by \eqref{LTC:extended}, we have the following:

\begin{theorem}\label{left:trivialized:EL:theorem}
	Let $G$ be a Lie group and $L:TG\times \R\rightarrow \R$ a contact Lagrangian function. If we introduce the coordinates $(g,\xi,z)$ on $TG\times \R$ given by \eqref{LTC:extended}, then the curves $\xi$ and $z$ satisfy the \textit{Herglotz equations on Lie groups}
		\begin{equation}\label{Herglotz:LG}
			\begin{split}
				\frac{d}{dt}\frac{\delta L}{\delta \xi} & =\emph{ad}_{\xi}^{*}\frac{\delta L}{\delta \xi}+T_{e}^{*}\mathcal{L}_{g}\left( \frac{\delta L}{\delta g} \right)+\frac{\delta L}{\delta \xi}\frac{\partial L}{\partial z} \\
				\dot{g} & = T_{e}\mathcal{L}_{g}\left( \xi \right) \\
				\dot{z} & = L.
			\end{split}
		\end{equation}
\end{theorem}

\begin{proof}
	Before starting the proof, let $\tilde{L}:G\times \mathfrak{g}\times \R \rightarrow \R$ be the left trivialized Lagrangian function using the coordinates \eqref{LTC:extended}. So that
	\begin{equation*}
		\tilde{L}(g,\xi,z)=L(g,g\xi,z),
	\end{equation*}
	where we are using the notation $g\xi:=T_{e}\mathcal{L}_{g}(\xi)$. In what follows, we will find critical values of the action along a trajectory on $TG\times \R$ of the form $(g(t),\dot{g}(t),z(t))$. But, under the left trivialized coordinates, this is equivalent to find critical values along trajectories on $G\times \mathfrak{g}\times \R$ of the form $(g(t),\xi(t),z(t))$, where
	\begin{equation*}
		\xi(t):=T_{g(t)}\mathcal{L}_{g(t)^{-1}}(\dot{g}(t)):=g(t)^{-1} \cdot \dot{g}(t).
	\end{equation*}
	
	Also, note that for any variation of a solution $g:I\rightarrow G$ of Euler-Lagrange equations denoted by $c:U\subset \R^{2}\rightarrow G$, the tangent vectors $\displaystyle{\frac{\partial c}{\partial s}}$ and $\displaystyle{\frac{\partial c}{\partial t}}$ may be transported to the Lie algebra $\mathfrak{g}$ in the following way
	\begin{equation}\label{Lie:algebra:variation}
		\begin{split}
			\eta:U & \rightarrow \mathfrak{g} \\
			(t,s) & \mapsto T_{c(t,s)}\mathcal{L}_{c(t,s)^{-1}}\left(\frac{\partial c}{\partial s}\right)
		\end{split}
		\quad
		\begin{split}
			\xi:U & \rightarrow \mathfrak{g} \\
			(t,s) & \mapsto T_{c(t,s)}\mathcal{L}_{c(t,s)^{-1}}\left(\frac{\partial c}{\partial t}\right).
		\end{split}
	\end{equation}
	Moreover, it is a standard fact (see \cite{bloch1996euler, holm1998euler}) that these maps satisfy
	\begin{equation}\label{Lie:algebra:variation:eq}
		\frac{\partial \xi}{\partial s}-\frac{\partial \eta}{\partial t}=[\xi,\eta].
	\end{equation}
	Conversely, if there are 2-parameter family of curves $\xi$ and $\eta$ on the Lie algebra satisfying \eqref{Lie:algebra:variation:eq}, there exists a smooth function $c:U\subset \R^{2}\rightarrow G$ such that $\xi$ and $\eta$ are given by 
 equations \eqref{Lie:algebra:variation}.
	
	We will often commit a slight abuse of notation with the letters $\xi$ and $\eta$, in the sense that they represent both a curve and a variation of curves in $\mathfrak{g}$. The first will be denoted by $\xi(t)$ and the second by $\xi(t,s)$, respectively $\eta(t)$ and $\eta(t,s)$. The relation between them is that
	\begin{equation}\label{xi:eta}
		\xi(t)=\xi(t,0)=g^{-1}\cdot \dot{g} \quad \text{and} \quad \eta(t)=\eta(t,0):=g^{-1}\cdot \delta g
	\end{equation}
	
	Now, find equations \eqref{Herglotz:LG} are equivalent to Herglotz equations on the contact extended space $TG\times \R$ let us show that its solutions satisfy Herglotz principle. Let the action over a curve be given by
	\begin{equation*}
		\mathcal{A}(g(\cdot),z_{0})=\int_{0}^{h}L(g(t),\dot{g}(t),z(t)) \ dt,
	\end{equation*}
	where the curve $z(t)$ satisfies the equation
	\begin{equation*}
	\dot{z}=L(g(t),\dot{g}(t),z(t)), \quad z(0)=z_{0}.
	\end{equation*}
	
	The action over a variation $c$ of the curve $g$ might be written as
	\begin{equation*}
		\mathcal{A}(c(\cdot,s),z_{0}) = \int_{0}^{h} L \left(c(t,s),\frac{\partial c}{\partial t}(t,s),z(t,s) \right) \ dt,
	\end{equation*}
	which in the coordinates \eqref{LTC:extended} is just
	\begin{equation*}
		\mathcal{A}(c(\cdot,s),z_{0}) = \int_{0}^{h} \tilde{L} \left(c(t,s),\xi(t,s),z(t,s) \right) \ dt,
	\end{equation*}
	where $\xi$ is defined as in \eqref{Lie:algebra:variation} and the variation $z(t,s)$ satisfies the equation
	\begin{equation*}
	\dot{z}(t,s)=\tilde{L}(c(t,s),\xi(t,s),z(t,s)), \quad z(0, s)=z_{0}.
	\end{equation*}
	
	
	
	Taking variations $\xi(t,s)$ and $\eta(t,s)$ as in \eqref{Lie:algebra:variation} and \eqref{Lie:algebra:variation:eq}, the first variation of the action functional gives
	\begin{equation*}
		\begin{split}
			\left. \frac{d}{ds} \right|_{s=0} \int_{0}^{h}  & \tilde{L} \left(c(t,s),\xi(t,s),z(t,s) \right)  \ dt \\
			& = \int_{0}^{h}\left[\frac{\partial \tilde{L}}{\partial g}\delta g + \frac{\partial \tilde{L}}{\partial \xi}\delta \xi + \frac{\partial \tilde{L}}{\partial z}\delta z\right] \ dt = z(h, 0).
		\end{split}
	\end{equation*}
	Observe now that the infinitesimal variation of the curve $z$, that is, $\delta z$ must satisfy the inhomogeneous linear differential equation
	\begin{equation*}
		\frac{\partial \delta z}{\partial t}(t)=\frac{\partial \tilde{L}}{\partial g}\delta g + \frac{\partial \tilde{L}}{\partial \xi}\delta \xi + \frac{\partial \tilde{L}}{\partial z}\delta z.
	\end{equation*}
	This is a first order ordinary differential equation whose solution is
	\begin{equation*}
		\delta z(\tau)=e^{\int_{0}^{\tau}\frac{\partial \tilde{L}}{\partial z}ds}\left( \int_{0}^{\tau}e^{-
			\int_{0}^{t}\frac{\partial \tilde{L}}{\partial z}ds}\left[\frac{\partial \tilde{L}}{\partial g}\delta g + \frac{\partial \tilde{L}}{\partial \xi}\delta \xi\right] \ dt \right),
	\end{equation*}
	since $\delta z(0)$ vanishes. Now, observe that we have seen above that the first variation of the action is equal to $\delta z(h)$. Hence,
	\begin{equation*}
		\left. \frac{d}{ds} \right|_{s=0}\mathcal{A}(c(\cdot,s),z_{0})= e^{\int_{0}^{h}\frac{\partial \tilde{L}}{\partial z}ds}\left( \int_{0}^{h}e^{-
			\int_{0}^{t}\frac{\partial \tilde{L}}{\partial z}ds}\left[\frac{\partial \tilde{L}}{\partial g}\delta g + \frac{\partial \tilde{L}}{\partial \xi}\delta \xi\right] \ dt \right),
	\end{equation*}
	where from \eqref{Lie:algebra:variation:eq} we have that $\delta \xi=\dot{\eta}+[\xi,\eta]$ and from \eqref{xi:eta} we have that $\delta g = g \cdot \eta$. Defining the auxiliary function
	\begin{equation*}
		\sigma(t)=e^{-
			\int_{0}^{t}\frac{\partial \tilde{L}}{\partial z}ds},
	\end{equation*}
	we obtain
	\begin{equation*}
		\left. \frac{d}{ds} \right|_{s=0}\mathcal{A}(c(\cdot,s),z_{0})= \frac{1}{\sigma(h)}\left( \int_{0}^{h}\sigma(t) \left[\frac{\partial \tilde{L}}{\partial \xi} (\dot{\eta}+\text{ad}_{\xi}\eta)+\frac{\partial \tilde{L}}{\partial g}g \cdot \eta \right] \ dt \right).
	\end{equation*}
	Integrating by parts, in order to get rid of time derivatives of $\eta$, and noting that the boundary terms vanish since $\eta$ is zero at the endpoints we conclude that the action reduces to the expression below
	\begin{equation*}
		\frac{1}{\sigma(h)}\left( \int_{0}^{h} 	\left[\sigma(t)\text{ad}_{\xi}^{*}\frac{\partial \tilde{L}}{\partial \xi}-\frac{d}{dt}\left(\sigma(t)\frac{\partial \tilde{L}}{\partial \xi}\right)+\sigma(t)T_{e}^{*}\mathcal{L}_{g}\left(\frac{\partial \tilde{L}}{\partial g}\right)\right]\eta \ dt \right).
	\end{equation*}
	Since $\eta$ is arbitrary and due to the fundamental theorem of calculus of variations, the curve $(g(t),\xi(t),z(t))$ is a critical point of the action if and only if it satisfies the equation
	\begin{equation*}
		\sigma(t)\text{ad}_{\xi}^{*}\frac{\partial \tilde{L}}{\partial \xi}-\sigma(t)\frac{d}{dt}\left(\frac{\partial \tilde{L}}{\partial \xi}\right)-\dot{\sigma}\frac{\partial \tilde{L}}{\partial \xi}+\sigma(t)T_{e}^{*}\mathcal{L}_{g}\left(\frac{\partial \tilde{L}}{\partial g}\right)=0,
	\end{equation*}
	which finishes the proof since this is equivalent to
	\begin{equation*}
		\sigma(t)\left[\text{ad}_{\xi}^{*}\frac{\partial \tilde{L}}{\partial \xi}-\frac{d}{dt}\left(\frac{\partial \tilde{L}}{\partial \xi}\right)+\frac{\partial \tilde{L}}{\partial z}\frac{\partial \tilde{L}}{\partial \xi}+T_{e}^{*}\mathcal{L}_{g}\left(\frac{\partial \tilde{L}}{\partial g}\right)\right]=0
	\end{equation*}
	and the auxiliary function $\sigma$ is nowhere zero. The last two equations in \eqref{Herglotz:LG} follow by construction.
\end{proof}

Next, let us define first the \textit{lifted action of $G$ on $TG\times \R$} whose translation by $g\in G$ is the map $\hat{\mathcal{L}}_{g}:TG\times \R \rightarrow TG\times \R$ given by
\begin{equation}\label{lhat}
\hat{\mathcal{L}}_{g}(v_{h},z)=(T_{h}\mathcal{L}_{g}(v_{h}),z).
\end{equation}

Suppose that $L:TG\times \R \rightarrow \R$ is a \textit{left invariant} contact Lagrangian function, thus satisfying $L\circ \hat{\mathcal{L}}_{g}=L$. We define the \textit{reduced contact Lagrangian} to be the function $l:\mathfrak{g}\times \R \rightarrow \R$ which is the restriction of the contact Lagrangian function $L$ to the space $T_{e}G \times \R$.

Observe that in this case
\begin{equation*}
	l(\xi,z)=L(e,\xi,z)=L(g,g\cdot \xi,z)
\end{equation*}
by the left-invariant property. The results of the previous theorem are further simplified using this additional assumption and we obtain in consequence the \textit{Euler-Poincaré-Herglotz} equations on the product space $\mathfrak{g}\times \R$.

\begin{theorem}\label{4points}
	Let $G$ be a Lie group and $L:TG\times \R\rightarrow \R$ a left-invariant contact Lagrangian function. If $l:\mathfrak{g}\times \R \rightarrow \R$ is the corresponding reduced Lagrangian, $g:I\rightarrow G$ is a curve on the Lie group and $\xi:I\rightarrow \mathfrak{g}$ a  curve on the Lie algebra, then the following are equivalent:
	\begin{enumerate}
		\item The curves $g$ and $z:I\rightarrow\R$ satisfy Herglotz equations for $L$;
		\item The Herglotz principle
		\begin{equation*}
		\delta \int_{0}^{h} L(g(t),\dot{g}(t),z(t)) \ dt=0, \quad \dot{z}=L(g(t),\dot{g}(t),z(t))
		\end{equation*}
		holds for variations of $g$ with fixed endpoints;
		\item The curves $\xi(t)=T_{g(t)}\mathcal{L}_{g^{-1}(t)}(\dot{g}(t))$ and $z$ satisfy the \textit{Euler-Poincaré-Herglotz} equations
		\begin{equation}\label{Herglotz:Poincare}
		\frac{d}{dt}\frac{\delta l}{\delta \xi}=\emph{ad}_{\xi}^{*}\frac{\delta l}{\delta \xi}+\frac{\delta l}{\delta \xi}\frac{\partial l}{\partial z} \quad \text{and} \quad \dot{z}=l.
		\end{equation}
		\item The reduced Herglotz variational principle
		\begin{equation}\label{reduced:principle}
		\delta \int_{0}^{h} l(\xi(t),z(t)) \ dt=0, \quad \dot{z}=l(\xi(t),z(t))
		\end{equation}
		holds using variations of the form $\delta \xi = \dot{\eta}+[\xi,\eta]$, where $\eta$ is a curve on $\mathfrak{g}$ vanishing at the endpoints.
	\end{enumerate}
\end{theorem}

\begin{proof}
	The equivalence of items $1.$ and $2.$ is general for all contact manifolds, so in particular for $TG\times \R$ (see \cite{anahory2021geometry}).
	
	Next, we prove the equivalence of both variational principles, that is, items $2.$ and $4.$ Observe that $l:\mathfrak{g}\times\mathbb{R}\to\mathbb{R}$ determines uniquely the function $L:TG\times\mathbb{R}\to\mathbb{R}$ by left-translations and viceversa. Thus one only needs to show that variations $(\delta g,\delta z)\in TG\times\mathbb{R}$ with $\delta g$ having fixed endpoints induces and are induced by variations $(\delta\xi,\delta z)$ with $\delta\xi$ of the form $\delta\xi=\dot{\eta}+[\xi,\eta]$, where $\eta(t)$ vanishes at end points. But, this is the content in the proof of Theorem \ref{left:trivialized:EL:theorem}.

 To conclude, we show the equivalence between $3$. and $4$. Indeed by using the definition of the auxiliary function $\sigma(t)$, the calculations in the proof of Theorem \ref{left:trivialized:EL:theorem} and integration by parts, we deduce that the reduced Herglotz variational principle holds if and only if
\begin{equation}
\sigma(t)\left[\text{ad}_{\xi}^{*}\frac{\partial l}{\partial \xi}-\frac{d}{dt}\left(\frac{\partial l}{\partial \xi}\right)+\frac{\partial l}{\partial z}\frac{\partial l}{\partial \xi}\right] = 0.
\end{equation}
So, since the auxiliary function $\sigma$ is nowhere zero, the result follows.
\end{proof}

\begin{remark}
	If we have determined a trajectory $(\xi(t),z(t))$ in $\mathfrak{g}\times \R$ solving Euler-Poincaré-Herglotz equations, we may obtain the corresponding trajectory on the original phase space $TG\times \R$ through the \textit{reconstruction} procedure, which amounts to solving the reconstruction equation
	\begin{equation*}
		\dot{g}=g \cdot \xi.
	\end{equation*}
\end{remark}

\begin{example}\label{so3:example}
	Consider now the matrix Lie group $SO(3)$, and a Lagrangian function $L:TSO(3)\times \R \rightarrow \R$ of the form
	\begin{equation*}
		L(R,\dot{R},z)=\frac{1}{2}\langle\langle \dot{R},\dot{R} \rangle\rangle  - \gamma z,
	\end{equation*}
with
	$$\langle\langle \dot{R},\dot{R} \rangle\rangle = \int_{B}\rho(X)\|\dot{R}X\|^{2} \ d^{3}X,$$
	where $B$ is the subset of $\R^{3}$ occupied by some rigid body. In this case, the left multiplication map is defined by $\mathcal{L}_{R_{1}}:SO(3)\rightarrow SO(3)$
	\begin{equation*}
		\mathcal{L}_{R_{1}}(R)=R_{1}R
	\end{equation*}
	and so its tangent map is just
	\begin{equation*}
		T_{R}\mathcal{L}_{R_{1}}(Y)=R_{1}Y, \quad Y\in T_{R}SO(3).
	\end{equation*}
	Moreover, the Lie algebra of $SO(3)$, denoted by $\mathfrak{so}(3)$ is composed by $3\times 3$ skew-symmetric matrices and the adjoint map is given by matrix commutator, i.e.,
	$$\text{ad}_{\xi} \eta =\xi \eta-\eta \xi, \quad \xi,\eta \in \mathfrak{so}(3).$$
	
	Now, the expression of $L$ on left trivialized coordinates is well-known to be
	$$L(R,\xi,z)=\frac{1}{2}\langle \xi,\xi \rangle  - \gamma z,$$
	where the inner product on $\mathfrak{so}(3)$ is given by
	$$\langle \xi_{1},\xi_{2} \rangle = \frac{1}{2}\text{tr} \ (\xi_{1}^{T}\mathbb{I}\xi_{2}),$$
	with $\mathbb{I}$ the inertia tensor. Many of the computation are easily carried using the hat map $\hat{(\cdot)}:\R^{3}\rightarrow \mathfrak{so}(3)$ defined by
	\begin{equation*}
	(\omega_{1},\omega_{2},\omega_{3}) \mapsto \left(\begin{matrix}
	0 & \omega_{1} & \omega_{2} \\
	-\omega_{1} & 0 & \omega_{3} \\
	-\omega_{2} & -\omega_{3} & 0
	\end{matrix}\right)
	\end{equation*}
	which is a Lie algebra isomorphism. Under these definitions the Lagrangian function is written as
	$$L(R,\xi,z)=\frac{1}{2}\xi^{T}\mathbb{I}\xi - \gamma z$$
	and the equations \eqref{Herglotz:LG} give
	\begin{equation*}
		\begin{split}
			\mathbb{I}\dot{\xi} & =-\xi \times \mathbb{I}\xi - \gamma \mathbb{I} \xi \\
			\dot{\xi} & = R\xi \\
			\dot{z} & = L.
		\end{split}
	\end{equation*}	\hfill$\diamond$
\end{example}

\section{Lie-Poisson-Jacobi reduction}

In this section, we will introduce a Jacobi structure on the space $\mathfrak{g}^{*}\times \R$ and we will see that for a regular $G$-invariant Lagrangian function the reduced dynamics is Hamiltonian with respect to this Jacobi structure. 

\subsection{Contact dynamics on a Lie group}

Now, assume that the configuration space $Q$ is a Lie group $G$ and that the regular Lagrangian function $L:TG\times\mathbb{R}\to\mathbb{R}$ is $G$-invariant. Then, similarly to what happens between the contact manifold $TG\times \R$ and the canonical contact manifold $T^{*}G\times \R$, we may also identify $\mathfrak{g}\times \R$ and $\mathfrak{g}^{*}\times \R$ through the \textit{reduced Legendre transformation} $\F l:\mathfrak{g}\times \R \rightarrow \mathfrak{g}^{*}\times \R$ given by
\begin{equation*}
	\F l(\xi,z)=\left( \frac{\delta l}{\delta \xi},z \right).
\end{equation*}
Note that,
\begin{equation*}
	\F L (e, \xi,z)=(e,\F l(\xi,z)),
\end{equation*}
that is, the reduced Legendre transform is essentially the restriction of the Legendre transform to $\mathfrak{g}\times \R$. Thus, the requiring $L$ to be regular is sufficient in order to make $\F l$ a local diffeomorphism.

\begin{proposition}
	If $L$ is $G-$invariant, the Legendre transform is $G$-equivariant, that is,
	\begin{equation*}
		\F L(g\cdot v, z)=g \cdot \F L(v, z), \quad (v,z) \in TG\times \R.
	\end{equation*}
\end{proposition}

Suppose that the Legendre transformation is a diffeomorphism and consider the function $h:\mathfrak{g}^{*}\times \R \rightarrow \R$ given by
\begin{equation*}
	h(\mu,z)=\langle \mu, \xi \rangle-l(\xi,z),
\end{equation*}
where $\displaystyle{\mu=\frac{\delta l}{\delta \xi}}$. Then it is easy to see that
\begin{equation*}
	\frac{\delta h}{\delta \mu}=\xi.
\end{equation*}

By inserting these equalities in Euler-Poincar\'e-Herglotz equations we get the following result:

\begin{proposition}
Given a regular reduced Lagrangian function $l$, if $(\xi(t),z(t))$ is a solution of Euler-Poincaré-Herglotz equations and $\mu(t)=\frac{\delta l}{\delta \xi}(\xi(t),z(t))$ then the curve $t \mapsto (\mu(t),z(t))$ satisfies
	\begin{equation*}
		\begin{split}
			\dot{\mu} & =\text{ad}^{*}_{\frac{\delta h}{\delta \mu}}\mu-\mu\frac{\partial h}{\partial z} \\
			\dot{z} & = \langle \mu,\frac{\delta h}{\delta \mu} \rangle-h(\mu,z).
		\end{split}
	\end{equation*}
\end{proposition}

\begin{proof}
  By differentiating the curve $t\mapsto \mu(t)$ we get
    $$\dot{\mu} = \frac{d}{dt}\frac{\delta l}{\delta \xi} = ad^{*}_{\xi} \frac{\delta l}{\delta \xi} + \frac{\delta l}{\delta \xi} \frac{\partial l}{\partial z},$$
    where we used Euler-Poincar{\'e}-Herglotz equations in the last equality. Using the equalities preceding the proposition statement, we deduce
    $$\dot{\mu} = \text{ad}^{*}_{\frac{\delta h}{\delta \mu}}\mu-\mu\frac{\partial h}{\partial z}.$$ The equation for $\dot{z}$ follows from the definition of $h$.
\end{proof}

These equations will be called \textit{Lie-Poisson-Jacobi equations}. In the next subsections, we will show that this is a Hamiltonian system with respect to a Jacobi structure on $\mathfrak{g}^{*}\times \R$.

\begin{example}
	Let us consider again the matrix Lie group $SO(3)$ and consider the hyperregular $G$-invariant Lagrangian function
	$$L(R,\dot{R},z) = \frac{1}{2}\langle\langle \dot{R},\dot{R} \rangle\rangle - \gamma z,$$
	so that we might consider the reduced Lagrangian which is the restriction of $L$ to $\mathfrak{g}\times \R$
	$$l(\xi,z)=\frac{1}{2}\xi^{T}\mathbb{I}\xi-\gamma z.$$
	Under all the assumptions taken in Example \ref{so3:example}, the Legendre transform is given by the map
	\begin{equation*}
		\F l(\xi,z)=(\mathbb{I}\xi,z).
	\end{equation*}
	Suppose that $\mathbb{I}$ is a positive definite symmetric matrix, so that $\F l$ is a diffeomorphism and we might define the Hamiltonian function
	$$h(\mu,z)= \mu^{T} \mathbb{I}^{-1}\mu -l(\mathbb{I}^{-1}\mu,z)=\mu^{T} \mathbb{I}^{-1}\mu -\frac{1}{2}(\mathbb{I}^{-1}\mu)^{T}\mu + \gamma z,$$
	and letting $\mathbb{J}=(2\mathbb{I}^{-1}-(\mathbb{I}^{-1})^{T})=\mathbb{I}^{-1}$ we get
	$$h(\mu,z)= \frac{1}{2}\mu^{T}\mathbb{J}\mu + \gamma z.$$
	The Lie-Poisson-Jacobi equations are then given by
	\begin{equation*}
		\begin{split}
			\dot{\mu} & =\text{ad}^{*}_{\frac{\delta h}{\delta \mu}}\mu-\gamma\mu \\
			\dot{z} & = \frac{1}{2}\mu^{T}\mathbb{J}\mu - \gamma z.
	\end{split}
	\end{equation*}\hfill$\diamond$
\end{example}

\subsection{Lie-Poisson-Jacobi bracket}

Let us recall first the definition of a Jacobi structure {(see \cite{Kirillov1976} and \cite{Lichnerowicz1978}).

\begin{definition}(Jacobi structure)
	A \textit{Jacobi structure} on a manifold $M$ is a tuple $(\Lambda,E)$, where $\Lambda$ is a bi-vector field and $E$ is a vector field, satisfying the following equations
	\begin{equation*}
		[\Lambda,\Lambda]=2E\wedge \Lambda, \quad [E,\Lambda]=0,
	\end{equation*}
	with $[\cdot,\cdot]$ the Schouten-Nijenhuis bracket.
\end{definition}

\begin{definition}(Jacobi bracket)
	A Jacobi bracket on a manifold $M$ is a bilinear, skew-symmetric map $\{\cdot,\cdot\}:C^{\infty}(M)\times C^{\infty}(M) \rightarrow C^{\infty}(M)$ satisfying the Jacobi identity and the following weak Leibniz rule
	\begin{equation*}
		\text{supp}(\{f,g\})\subseteq \text{supp}(f)\cap \text{supp}(g).
	\end{equation*}
\end{definition}

A \textit{Jacobi manifold} is a manifold possessing either a Jacobi structure or a Jacobi bracket since these two definitions are equivalent (see \cite{Kirillov1976}, \cite{Lichnerowicz1978}, \cite{marle}, \cite{ibanez1997co}). However, it is much more convenient to introduce a Jacobi structure for practical purposes.

Now, from the Jacobi structure we can define an associated Jacobi bracket as follows: 
\[
\{f, g\}=\Lambda(df, dg)+f E(g)-g E(f), \quad f, g\in C^{\infty}(M, \R)
\]
In this case, the weak Leibniz rule is equivalent to the generalized Leibniz rule
\begin{equation}
\{f, gh\} = g\{f, h\} + h\{f, g\} + ghE(f),
\end{equation}
In this sense, this bracket generalizes the well-known Poisson brackets. Indeed, a Poisson manifold is a particular case of Jacobi manifold in which $E=0$. 

Given a Jacobi manifold $(M,\Lambda,E)$, we consider the map $\sharp_{\Lambda}:\Omega^{1}(M)\rightarrow \mathfrak{X}(M)$ defined by
\begin{equation*}
\sharp_{\Lambda}(\alpha)=\Lambda(\alpha,\cdot).
\end{equation*}
We have that $\sharp_{\Lambda}$ is a morphism of $C^{\infty}$-modules, though it may fail to be an isomorphism. Given a function $f: M \rightarrow \R$,  we define the Hamiltonian vector field $X_f$ by
\[
X_f=\sharp_{\Lambda} (df) + f E.
\]

Contact structures are examples of Jacobi structures. Given a contact manifold $(M,\eta)$, we may associate it a natural Jacobi structure. Indeed, we define the bivector $\Lambda$ as  
\begin{equation}\label{Lambda:intrinsic}
\Lambda(\alpha, \beta)=-d\eta(\flat^{-1}(\alpha), \flat^{-1}(\beta)), \qquad \alpha, \beta \in \Omega^1(M)\; .
\end{equation}
So that the pair $(\Lambda, E=-\Reeb)$ is a Jacobi structure (see \cite{Lichnerowicz1978}, \cite{de2017cosymplectic}).

In Darboux coordinates, the bivector $\Lambda$ reads as
\begin{equation}\label{Lambda:coordinates}
\Lambda=\frac{\partial}{\partial p_i}\wedge \left(\frac{\partial}{\partial q^i}+p_i\frac{\partial}{\partial z}\right).
\end{equation}

In addition, given a function $h:M\rightarrow \R$, the \textit{contact Hamiltonian vector field} is given by
\begin{equation*}
	X_{h}=\sharp_{\Lambda}(dh)-h \Reeb
\end{equation*}
and along its integral curves the following equation is satisfied with respect to the associated Jacobi bracket
\begin{equation}\label{hamiltonian:evolution}
\dot{f}=\{h,f\}-f\frac{\partial h}{\partial z}, \quad \forall f\in C^{\infty}(M).
\end{equation}

The preceding equation implies in particular that the Hamiltonian function is not conserved along its integral curves since
\begin{equation*}
\dot{h}=-h\frac{\partial h}{\partial z}.
\end{equation*}

In Darboux coordinates, the bracket is given by
\begin{eqnarray*}
	\{f,g\} = \frac{\partial f}{\partial q^i}\frac{\partial g}{\partial p_i}-\frac{\partial f}{\partial p_i}\frac{\partial g}{\partial q^{i}}
	-\frac{\partial f}{\partial z}\left(p_i\frac{\partial g}{\partial p_i}-g\right)+\frac{\partial g}{\partial z}\left(p_i\frac{\partial f}{\partial p_i}-f\right)
\end{eqnarray*}
and the Hamiltonian vector field is given in canonical coordinates by:
\[
X_h=\frac{\partial h}{\partial p_i}\frac{\partial}{\partial q^i}
-\left(\frac{\partial h}{\partial q^i}+p_i\frac{\partial h}{\partial z}\right)
\frac{\partial}{\partial p_i}+
\left(p_i\frac{\partial h}{\partial p_i}-h\right)\frac{\partial}{\partial z}
\]

So, given a manifold $Q$, the contact manifold $(T^{*}Q\times \R, \eta_{Q})$, where $\eta_{Q}$ is the contact form defined in \eqref{eq:cotangent_contact_structure}, has a canonical Jacobi structure and an associated Jacobi bracket defined as above, which we will denote by $\{\cdot,\cdot\}_{can}$ from now on.

Let us now introduce a Jacobi structure on $\mathfrak{g}^{*}\times \R$. 

\begin{proposition}\label{prop4.6}
	The structure given by
	\begin{equation}
		\Lambda(\mu,z)(df,dg)=\left\langle \mu,\left[ \frac{\delta f}{\delta \mu},\frac{\delta g}{\delta \mu} \right] \right\rangle+ \left\langle \mu,\frac{\delta f}{\delta \mu} \right\rangle \frac{\partial g}{\partial z} - \left\langle \mu,\frac{\delta g}{\delta \mu} \right\rangle \frac{\partial f}{\partial z}
	\end{equation}
	and $\displaystyle{E = -\Reeb=-\frac{\partial}{\partial z}}$ is a Jacobi structure on $\mathfrak{g}^{*}\times \R$ and induces a Jacobi bracket given by
	\begin{equation}
		\{f,g\}(\mu, z) =\Lambda(\mu)(df,dg)-f\frac{\partial g}{\partial z}+g\frac{\partial f}{\partial z}.
	\end{equation}
\end{proposition}

\begin{proof}
	We just have to prove that $\Lambda$ is a Jacobi structure, since then the map $\{\cdot,\cdot\}$ would be the associated Jacobi bracket. Using the fact that
	\begin{equation*}
		\Lambda= \text{pr}_{1}^{*}\Lambda_{0} +\Delta \wedge \Reeb,
	\end{equation*}
	where $\Lambda_{0}$ is the Lie-Poisson structure on $\mathfrak{g}^{*}$ defined by
 $$\Lambda_{0}(\mu)(df, dg) = \left\langle \mu,\left[ \frac{\delta f}{\delta \mu},\frac{\delta g}{\delta \mu} \right] \right\rangle, \quad f,g \in C^{\infty}(\mathfrak{g}^{*}),$$
 $\text{pr}_{1}:\mathfrak{g}^{*}\times \R \rightarrow \mathfrak{g}^{*}$ is the projection onto the first factor, $\Delta$ is the vector field on $\mathfrak{g}^{*}\times \R$ defined by
	\begin{equation*}
		\Delta(\mu,z) = \left. \frac{d}{dt} \right|_{t=0}(\mu+t\mu,z)
	\end{equation*}
	and $\Reeb=\frac{\partial}{\partial z}$, we may deduce after some computations involving the Schouten-Nijenhuis bracket and interior products that (see \cite{marle1997schouten})
	\begin{equation*}
		[\Lambda,\Lambda] = 2[\Lambda_{0},\Delta \wedge \Reeb] =2\Lambda_{0}\wedge \Reeb,
	\end{equation*}
	where we used that $[\Delta \wedge \Reeb, \Lambda_{0}]=[\Lambda_{0}, \Delta \wedge \Reeb]$ (since $\Lambda_{0}$ and $\Delta \wedge \Reeb$ are both $(2,0)$-tensors) and the fact that $[\Lambda_0,\Delta]=\Lambda_0$. The previous equality is equivalent to
	\begin{equation*}
		[\Lambda,\Lambda] = 2\Lambda\wedge \Reeb,
	\end{equation*}
    using linearity and skew-symmetry of the wedge product.

In addition, we also have that $[\Reeb,\Lambda]$ vanishes:
 $$[\Reeb,\Lambda] = [\Reeb, \Lambda_{0}] + [\Reeb, \Delta \wedge \Reeb].$$
 the first term vanishes since $\Lambda_{0}$ is pulled-back from $\mathfrak{g}^{*}$ and so $[\Reeb, \Lambda_{0}] = \mathcal{L}_{\Reeb} \Lambda_{0}=0$. The second term also vanishes since $[\Reeb, \Delta \wedge \Reeb] = [\Reeb, \Delta]\wedge \Reeb + \Delta \wedge [\Reeb, \Reeb] = [\Reeb, \Delta]\wedge \Reeb$ and the Lie bracket of $[\Reeb, \Delta]$ is zero.
Hence, $(\Lambda,E=-\Reeb)$ is indeed a Jacobi structure.
\end{proof}
\begin{remark}
   The Jacobi structure $(\Lambda= \text{pr}_{1}^{*}\Lambda_{0} +\Delta \wedge R$, $R=-\partial/\partial z)$ on $\mathfrak{g}^{*}\times\mathbb{R}$ is linear on the vector space $\mathfrak{g}^{*}\times\mathbb{R}\to\mathbb{R}$. In fact, it is a particular case of a more general class of linear Jacobi structures on vector bundles which were previously considered in \cite{iglesias2000some} (see also \cite{iglesias2001some, iglesiasphd}).
\end{remark}
\begin{definition}\label{LPJ-bracket-def}
    The Jacobi structure and bracket defined in the previous Proposition are called the \textit{Lie-Poisson-Jacobi structure} and bracket  on $\mathfrak{g}^{*}\times \R$, respectively.
\end{definition}

In the next proposition, we will define the Hamiltonian vector fields associated with this Jacobi structure and we will see what is the expression of the corresponding Hamilton equations. 

\begin{proposition} Let $(\Lambda, R)$  be the Jacobi structure on $\mathfrak{g}^{*}\times \R$ defined above. If $h:\mathfrak{g}^{*}\times \R \rightarrow \R$ is the reduced Hamiltonian function, the integral curves of the Hamiltonian vector field $X_{h}$ satisfy the Lie-Poisson-Jacobi equations
	\begin{equation*}
		\begin{split}
			\dot{\mu} & =\emph{ad}^{*}_{\frac{\delta h}{\delta \mu}}\mu-\mu\frac{\partial h}{\partial z} \\
			\dot{z} & = \langle \mu,\frac{\delta h}{\delta \mu} \rangle-h(\mu,z).
		\end{split}
	\end{equation*}
\end{proposition}

\begin{proof}
	Note that if $(\mu(t),z(t))$ is an integral curve of the contact Hamiltonian vector field $X_{h}$, then we have that for all function $f$
	\begin{equation*}
		\dot{f}= \langle df(\mu,z), (\dot{\mu},\dot{z}) \rangle = \left\langle \dot{\mu},\frac{\delta f}{\delta \mu} \right\rangle + \frac{\partial f}{\partial z}\dot{z}.
	\end{equation*}
	Also, using the definition of the Jacobi bracket we have that
	\begin{equation*}
		\{h,f\}(\mu, z)=\left\langle \mu,\left[ \frac{\delta h}{\delta \mu},\frac{\delta f}{\delta \mu} \right] \right\rangle+ \left( \left\langle \mu,\frac{\delta h}{\delta \mu} \right\rangle-h \right) \frac{\partial f}{\partial z} - \left( \left\langle \mu,\frac{\delta f}{\delta \mu} \right\rangle - f \right) \frac{\partial h}{\partial z}.
	\end{equation*}
	Thus using \eqref{hamiltonian:evolution} we deduce that
	\begin{equation*}
		\dot{f}= \{h,f\}(\mu, z)- f\frac{\partial h}{\partial z} =\left\langle \text{ad}_{\frac{\delta h}{\delta \mu}}^{*}\mu-\mu\frac{\partial h}{\partial z},\frac{\delta f}{\delta \mu} \right\rangle + \left( \left\langle \mu,\frac{\delta h}{\delta \mu} \right\rangle-h \right) \frac{\partial f}{\partial z}.
	\end{equation*}
	Since the function $f$ is arbitrary the result follows.
\end{proof}

The following result relates the new bracket on the reduced product space $\mathfrak{g}^{*}\times \R$ with the canonical Jacobi bracket on $T^{*}G\times \R$. Recall how the the left action of $G$ on itself lifts to the cotangent space $T^{*}G$.
The cotangent left action is given by the map
\begin{equation*}
	\phi_{g}:T^{*}G \rightarrow T^{*}G, \quad \alpha_{h}\in T_{h}^{*}G \mapsto T_{gh}^{*}\mathcal{L}_{g^{-1}}(\alpha_{h})\in T_{gh}^{*}G.
\end{equation*}

\begin{proposition}\label{Jacobi:bracket:correspondence}
	Let $F$ and $H$ be $G$-invariant functions on $T^{*}G\times \R$ and let $f$ and $h$ be their restrictions to $\mathfrak{g}^{*}\times \R$, respectively. Then
	\begin{equation*}
		\{ F,H \}_{can}(\alpha,z) = \{ f,h \} (T_{e}^{*}\mathcal{L}_{g}(\alpha),z), \quad \alpha \in T_{g}^{*}G,
	\end{equation*}
	where $\{\cdot,\cdot\}_{can}$ is the canonical Jacobi bracket on $T^{*}G\times \R$ and $\{\cdot,\cdot\}$ is the Lie-Poisson-Jacobi structure on $\mathfrak{g}^{*}\times \R$.
\end{proposition}

\begin{proof}
	The canonical Jacobi bracket on $T^{*}G\times \R$ has the following form
	\begin{equation}\label{other:form}
		\{ F,H \}_{can} = \{ F\circ i_{z},H\circ i_{z} \}_{T^{*}G}-F\frac{\partial H}{\partial z}+H\frac{\partial F}{\partial z},
	\end{equation}
	where $\{\cdot,\cdot\}_{T^{*}G}$ is the canonical Poisson bracket on $T^{*}G$ and $i_{z}:T^{*}G \hookrightarrow T^{*}G\times \R$ is the inclusion mapping $i_{z}(\alpha)=(\alpha,z)$. Observe that it is a standard fact that
	\begin{equation*}
		\{ F,H \}_{T^{*}G}(\alpha)=\{ f,h \}_{0} (T_{e}^{*}\mathcal{L}_{g}(\alpha)),
	\end{equation*}
	where $\{\cdot,\cdot\}_{0}$ is the Lie-Poisson bracket on $\mathfrak{g}^{*}$ (see \cite{marsden2013introduction}, for instance). Moreover, since $F$ and $H$ are invariant functions we have that also their derivative with respect to the contact variable $z$. Therefore, \eqref{other:form} equals
	\begin{align*}
		\{ F,H \}_{can} (\alpha,z)=& \{ f\circ i_{z},h\circ i_{z} \}_{0}(T_{e}^{*}\mathcal{L}_{g}(\alpha))\\&-f(\alpha,z)\frac{\partial h}{\partial z}(\alpha,z)+h(\alpha,z)\frac{\partial f}{\partial z}(\alpha,z),
	\end{align*}
	which is exactly the bracket $\{ f,h \} (T_{e}^{*}\mathcal{L}_{g}(\alpha),z)$ (see Proposition \ref{prop4.6}).
\end{proof}


\begin{example}
    The reduced space of $T^{*}SO(3)\times \R$ is $\mathfrak{so}(3)^{*}\times \R$, which we can identify with $\R^{4}$ under the suitable Lie algebra isomorphism. The corresponding Lie-Poisson-Jacobi bracket is
    \begin{equation*}
        \begin{split}
            \{h,f\}(\mu,z) & = \mu \cdot \left( \frac{\delta h}{\delta \mu} \times \frac{\delta f}{\delta \mu} \right) + \left(\mu \cdot \frac{\delta h}{\delta \mu}   - h\right)\frac{\partial f}{\partial z} - \left(\mu \cdot \frac{\delta f}{\delta \mu}  - f\right)\frac{\partial h}{\partial z} \\
            & = \frac{\delta f}{\delta \mu} \cdot \left( \mu \times \frac{\delta h}{\delta \mu} - \mu \frac{\partial h}{\partial z}\right) + \left(\mu \cdot \frac{\delta h}{\delta \mu}   - h\right)\frac{\partial f}{\partial z} + f\frac{\partial h}{\partial z}
        \end{split}
    \end{equation*}
    Given the Hamiltonian function $h(\mu,z) = \frac{1}{2}\mu^{T}\mathbb{J}\mu + \gamma z$, the corresponding Lie-Poisson-Jacobi equations are then
    \begin{equation*}
        \begin{split}
            \dot{\mu} & = \mu \times \mathbb{J}\mu  - \gamma \mu \\
            \dot{z} & = \frac{1}{2}\mu^{T}\mathbb{J}\mu - \gamma z.
        \end{split}
    \end{equation*}
\end{example}

\subsection{Lie-Poisson-Jacobi reduction theorem}

The objective of this section is two-fold: we are going to prove that the standard momentum map of the Lie group $T^{*}G$ induces a Jacobi map from $T^{*}G\times \R$ to $\mathfrak{g}^{*}\times \R$. Then we prove that Jacobi maps allow to perform a reduction and to find a Jacobi structure on $\mathfrak{g}^{*}\times \R$. If we start with the canonical Jacobi structure then the reduced structure will be precisely the Lie-Poisson-Jacobi  structure we defined in the last section. Moreover, we will prove that Jacobi maps project Hamiltonian vector fields onto Hamiltonian vector fields. This is the Jacobi version of the Poisson reduction Theorem given in \cite{marsden2013introduction}.

It is well-known (see \cite{Abraham1978}, \cite{marsden2013introduction}) that the map $J:T^{*}G\rightarrow \mathfrak{g}^{*}$ defined by
\begin{equation}\label{LG:momentum}
	\langle J(\alpha_{g}),\xi \rangle=\langle T_{e}^{*}\mathcal{L}_{g}(\alpha_{g}),\xi \rangle, \quad \alpha_{g}\in T_{g}^{*} G,\,\, \xi\in \mathfrak{g},
\end{equation}
is a momentum map on the cotangent bundle of the Lie group. Moreover it is a Poisson map with respect to the standard Poisson structure and Lie-Poisson structure on $T^{*}G$ and $\mathfrak{g}^{*}$, respectively.

\begin{definition}
	Given two Jacobi manifolds $(M_{1},\Lambda_{1},E_{1})$ and $(M_{2},\Lambda_{2},E_{2})$, with Jacobi brackets $\{\cdot,\cdot\}_1$ and $\{\cdot,\cdot\}_2$, respectively, and a smooth map $\varphi:M_{1} \rightarrow M_{2}$, then
	 the map $\varphi$ is said to be a \textit{Jacobi map} if for every $f$ and $h$ in $C^{\infty}(M_{2})$ we have that
		\begin{equation*}
			\{ f \circ \varphi, h \circ \varphi \}_{1}= \{ f,h \}_{2} \circ \varphi.
		\end{equation*} 
	
\end{definition}

The following theorem establishes that in the presence of a Jacobi action, the Jacobi bracket reduces to a Jacobi bracket on the space of orbits of the action. 

\begin{theorem}[Jacobi reduction theorem by a Lie Group action]\label{reduction:theorem}
	Let $G$ be a Lie group acting on a Jacobi manifold $(P,\{\cdot,\cdot\})$ by Jacobi maps. Suppose that $P/G$ is a smooth manifold and the projection $\pi:P\rightarrow P/G$ is a submersion. Then, there is a unique Jacobi bracket on $P/G$ denoted by $\{\cdot,\cdot\}_{red}$ called the reduced Jacobi bracket such that $\pi$ is a Jacobi map.
	
	Moreover, suppose that $H:P\rightarrow \R$ is $G$-invariant and define $h:P/G\rightarrow \R$ by $H :=h \circ \pi$. If $\phi_{t}^{X_{H}}$ and $\phi_{t}^{X_{h}}$ are the Hamiltonian flows of $H$ and $h$, respectively, then they satisfy the equation
	\begin{equation*}
		\pi \circ \phi_{t}^{X_{H}}= \phi_{t}^{X_{h}}\circ \pi.
	\end{equation*}
\end{theorem}

\begin{proof}
	Let $f,h \in C^{\infty}(P/G)$. Notice that the Jacobi bracket of the functions $f\circ \pi$ and $h\circ \pi$ is also $G$-invariant. Indeed if $\phi_{g}:P\rightarrow P$ denotes the action then
	\begin{equation*}
		\{ f\circ \pi,h\circ \pi \} \circ \phi_{g} = \{f\circ \pi\circ \phi_{g},h\circ \pi\circ \phi_{g}\}=\{f\circ \pi,h\circ \pi\},
	\end{equation*}
	where the first equality holds since $\phi_{g}$ acts by Jacobi maps and the second holds since $\pi$ is invariant under composition with the action by definition. Now, observe that a function $\Phi$ on $P$ is $G$-invariant if and only if there exists a function $\varphi:P/G\rightarrow \R$ such that $\Phi=\varphi \circ \pi$. Thus, denote by $\{f,h\}_{red}$ the function on $P/G$ defined by
	\begin{equation*}
		\{ f\circ \pi,h\circ \pi \} = \{f,h\}_{red} \circ \pi.
	\end{equation*}
	We must check that $\{f,h\}_{red}$ is indeed a Jacobi bracket. It is clearly bilinear and skew-symmetric, since $\{f\circ \pi,h\circ \pi \}$ also is. It satisfies Jacobi identity since $\{\cdot,\cdot \}$ also does. And finally, we must check if it satisfies the weak Leibniz condition. But note that if $k \in C^{\infty}(P/G)$ and $(\Lambda,E)$ is the Jacobi structure associated to $\{\cdot,\cdot\}$
	\begin{align}
			 \{f,hk\}_{red} \circ \pi =& \{ f\circ \pi,(h\circ \pi)(k \circ \pi) \} \nonumber\\
			 =& (h\circ \pi)\{ f\circ \pi,k \circ \pi \} + (k \circ \pi)\{ f\circ \pi,h \circ \pi \} \label{eqproof4.13}\\&+ (h \circ \pi)(k \circ \pi)E(f \circ \pi)\nonumber
	\end{align}


Now, since $E(\Phi)=\{1,\Phi\}$, for $\Phi\in\mathcal{C}^{\infty}(P)$, we deduce that the vector field $E$ is $\pi$-projectable and, thus, there exists a vector field $E_{red}$ over $P/G$ such that $E_{red}(f)\circ \pi = E(f\circ \pi)$, $\forall$ $f\in\mathcal{C}^{\infty}(P/G)$. Therefore, from \eqref{eqproof4.13}, it follows that $$ \{f,hk\}_{red} \circ \pi=(h\{f,k\}_{red}+k\{f,h\}_{red}+hk E_{red}(f))\circ \pi.$$

	This implies that $\{\cdot,\cdot \}_{red}$ satisfies the weak Leibniz condition. In fact, following the same argument we could prove that the Jacobi structure associated to $\{\cdot,\cdot \}_{red}$, denoted by $(\Lambda_{red},E_{red})$, by uniqueness is given by
	\begin{equation*}
			\Lambda_{red}\circ \pi = T\pi (\Lambda) \quad \text{and} \quad E_{red}\circ \pi = T \pi(E).
	\end{equation*}
	Now, note that given $p\in P$
	\begin{equation*}
		\frac{d}{dt}(f\circ \pi \circ \phi_{t}^{X_{H}}(p))=\{H, f\circ \pi \}\circ \phi_{t}^{X_{H}}(p) + (f\circ \pi)E(H)\circ \phi_{t}^{X_{H}}(p).
	\end{equation*}
	Since $H=h\circ \pi$ the later is equivalent to
	\begin{equation}\label{hamiltonian:bracket:equation}
		\frac{d}{dt}(f\circ \pi \circ \phi_{t}^{X_{H}}(p))=\{h, f \}_{red}\circ\pi\circ \phi_{t}^{X_{H}}(p) + fE_{red}(h)\circ \pi \circ \phi_{t}^{X_{H}}(p).
	\end{equation}
	Similarly, on the quotient space $P/G$, letting $[p]:=\pi(p)$, we have that
	\begin{equation*}
		\frac{d}{dt}(f \circ \phi_{t}^{X_{h}}([p]))=\{h, f \}_{red}\circ \phi_{t}^{X_{h}}([p]) + fE_{red}(h)\circ \phi_{t}^{X_{h}}([p]).
	\end{equation*}
	This last equation, uniquely defines the Hamiltonian flow $\phi_{t}^{X_{h}}$. Therefore, by comparison with equation \eqref{hamiltonian:bracket:equation}, we must have
	\begin{equation*}
		\pi \circ \phi_{t}^{X_{H}}= \phi_{t}^{X_{h}}\circ \pi.
	\end{equation*}
\end{proof}
\begin{remark}
The first part of this theorem may be deduced of a more general construction for Jacobi manifolds in the literature (see \cite{joana1989} and also \cite{ibort1997reduction}). Anyway, we have included a proof of the result to make the paper more self-contained
\end{remark}
\begin{proposition}
	The smooth submersion $\hat{J}:T^{*}G \times \R \rightarrow \mathfrak{g}^{*}\times \R$ given by
	\begin{equation*}
		\hat{J}(\mu_{g},z)=(J(\mu_{g}),z)
	\end{equation*}
	is a Jacobi map between the canonical Jacobi manifolds $T^{*}G\times\mathbb{R}$ and $\mathfrak{g}^{*}\times\mathbb{R}$, where $J$ is given by \eqref{LG:momentum}.
\end{proposition}

\begin{proof}
	We must prove that for any $f$ and $h$ in $C^{\infty}(\mathfrak{g}^{*}\times \R)$ the following identity holds
	\begin{equation*}
		\{ f \circ \hat{J}, h \circ \hat{J} \}_{can}= \{ f,h \} \circ \hat{J},
	\end{equation*}
	where $\{\cdot,\cdot\}_{can}$ is the Jacobi bracket on $T^{*}G \times \R$ and $\{\cdot,\cdot\}$ is the Lie-Poisson-Jacobi bracket on $\mathfrak{g}^{*}\times \R$.
	
	Note that the function $f \circ \hat{J}$ is $G$-invariant for any function $f$. Indeed, given $\alpha_{h}\in T_{h}^{*}G$, we have that
	\begin{equation*}
		f \circ \hat{J} \circ \phi_{g}(\alpha_{h},z)=f \circ \hat{J}(T_{gh}^{*}\mathcal{L}_{g^{-1}}(\alpha_{h}),z)
	\end{equation*}
	where $\phi_{g}$ denotes the cotangent lifted action of $G$ on $T^{*}G \times \R$. Then applying the definition of $\hat{J}$, we deduce
	\begin{equation*}
		f \circ \hat{J} \circ \phi_{g}(\alpha_{h},z) = f (T_{e}^{*}\mathcal{L}_{gh}(T_{gh}^{*}\mathcal{L}_{g^{-1}}(\alpha_{h})),z).
	\end{equation*}
	Finally, applying the composition rule of cotangent maps we conclude
	\begin{equation*}
		f \circ \hat{J} \circ \phi_{g}(\alpha_{h},z) = f (T_{e}^{*}\mathcal{L}_{h}(\alpha_{h})),z) = f \circ \hat{J} (\alpha_{h},z).
	\end{equation*}
	Moreover, note that for any $f$ the restriction of $f\circ \hat{J}$ to the identity is just $f$ itself since
	\begin{equation*}
		f \circ \hat{J} (\alpha_{e},z)= f(T_{e}^{*}\mathcal{L}_{e}(\alpha_{e}),z)=f(\alpha_{e},z).
	\end{equation*}
	Hence, we may apply Proposition \ref{Jacobi:bracket:correspondence} to deduce that
	\begin{equation*}
		\{ f \circ \hat{J}, h \circ \hat{J} \}_{can} (\alpha_{h},z)= \{ f,h \} (T_{e}^{*}\mathcal{L}_{h}(\alpha_{h}),z)
	\end{equation*}
	which is equivalent to
	\begin{equation*}
		\{ f \circ \hat{J}, h \circ \hat{J} \}_{can} (\alpha_{h},z)= \{ f,h \} \circ \hat{J} (\alpha_{h},z).
	\end{equation*}
\end{proof}

\begin{theorem}
	Let $H:T^{*}G \times \R \rightarrow \R$ be a $G$-invariant function. Then the function $h:\mathfrak{g}^{*}\times \R\rightarrow \R$ given as $h:=H|_{\mathfrak{g}^{*}\times \R}$ satisfies
	\begin{equation*}
		H = h\circ \hat{J}.
	\end{equation*}
	Moreover, the contact Hamiltonian vector field $X_{H}$ on $T^{*}G\times \R$ with flow $\phi_{t}^{X_{H}}$ and the Jacobi Hamiltonian vector field $X_{h}$ on $\mathfrak{g}^{*}\times \R$ with flow $\phi_{t}^{X_{h}}$ are related by $\hat{J}$ or, equivalently, their flows satisfy the equation
	\begin{equation*}
		\hat{J} \circ \phi_{t}^{X_{H}}= \phi_{t}^{X_{h}}\circ \hat{J}.
	\end{equation*}
\end{theorem}

\begin{proof}
	First note that for any $(\alpha_{g},z)\in T^{*}G\times \R$ we have that
	\begin{equation*}
		H(\alpha_{g},z)=H\circ \hat{J} (\alpha_{g},z),
	\end{equation*}
	since $H$ is $G$-invariant. Then, since $\hat{J} (\alpha_{g},z)\in \mathfrak{g}^{*}\times \R$, we have that
	\begin{equation*}
		H(\alpha_{g},z)=h \circ \hat{J} (\alpha_{g},z),
	\end{equation*}
	which proves the first statement in the theorem.
	
	The second statement is a consequence of Theorem \ref{reduction:theorem}. Indeed, note that the map
	\begin{equation*}
		\begin{split}
			\psi: (T^{*}G\times \R)/G & \longrightarrow \mathfrak{g}^{*}\times \R \\
			[\alpha_{g},z] & \mapsto \hat{J}(\alpha_{g},z)
		\end{split}
	\end{equation*}
	is a diffeomorphism. Moreover, by construction we have that $\psi \circ \pi = \hat{J}$, where $\pi: T^{*}G\times \R \rightarrow (T^{*}G\times \R)/G$ is the quotient map. Now, the map $\psi$ and the reduced Jacobi brackets on $(T^{*}G\times \R)/G$ given by Theorem \ref{reduction:theorem} induce a Jacobi bracket $\{\cdot,\cdot\}_{*}$ on $\mathfrak{g}^{*}\times \R$ that makes $\psi$ a Jacobi map so that
	\begin{equation*}
		\{f,h\}_{*}\circ \psi = \{ f\circ \psi, h \circ \psi \}_{red}, \quad f,g \in C^{\infty}(\mathfrak{g}^{*}\times \R).
	\end{equation*}
	By definition of the reduced bracket, $\pi$ is also a Jacobi map so that
	\begin{equation*}
		\{f,h\}_{*}\circ \psi \circ \pi = \{ f\circ \psi \circ \pi, h \circ \psi \circ \pi \}_{can}, \quad f,g \in C^{\infty}(\mathfrak{g}^{*}\times \R),
	\end{equation*}
	which by construction gives
	\begin{equation*}
		\{f,h\}_{*}\circ \hat{J} = \{ f\circ \hat{J}, h \circ \hat{J} \}_{can}, \quad f,g \in C^{\infty}(\mathfrak{g}^{*}\times \R).
	\end{equation*}
	But by Proposition \ref{Jacobi:bracket:correspondence} and since $f\circ \hat{J}, h \circ \hat{J}$ are $G$-invariant functions on $T^{*}G\times \R$ we must have that
	\begin{equation*}
		\{f,h\}_{*}=\{f,h\},
	\end{equation*}
	where $\{\cdot,\cdot\}$ is the canonical bracket on $\mathfrak{g}^{*}\times \R$ defined before. So, we may conclude that
	\begin{equation*}
		\hat{J} \circ \phi_{t}^{X_{H}}=\psi \circ \pi \circ \phi_{t}^{X_{H}}=\psi \circ \phi_{t}^{X_{h_{red}}} \circ \pi,
	\end{equation*}
	where we used the last statement of Theorem \ref{reduction:theorem}, relating the Hamiltonian flow of $H$ with that of $h_{red}$, the function defined by $H:=h_{red}\circ \pi$. Now observe that, since $\psi$ is a Jacobi map and $h_{red}=h\circ \psi$, we have that
	$$h_{red}\circ \pi=h\circ \hat{J}=h\circ \psi \circ \pi,$$
	the following relation holds between Hamiltonian flows
	\begin{equation*}
		\psi \circ \phi_{t}^{X_{h_{red}}}=\phi_{t}^{X_{h}}\circ \psi.
	\end{equation*}
	Therefore, we deduce
	\begin{equation*}
		\hat{J} \circ \phi_{t}^{X_{H}}= \phi_{t}^{X_{h}}\circ \psi \circ \pi=\phi_{t}^{X_{h}}\circ \hat{J}.
	\end{equation*}
\end{proof}

\begin{example}
	
    Consider the Lie group $SO(3)$ where the left action is given by $\mathcal{L}_{R}(R_{0})=RR_{0}$ for any $R,R_{0}\in SO(3)$. Therefore, $T_{e}\mathcal{L}_{R}(\xi)=R\xi$ for $\xi \in \mathfrak{so}(3)$ and $T_{e}^{*}\mathcal{L}_{R}(\mu_{R})=R^{T}\mu_{R}$, where $\mu_{R}\in T_{R}^{*}SO(3)$. Thus, we introduce the Jacobi map $\hat{J}:T^{*}SO(3)\times \R \rightarrow \mathfrak{so}(3)^{*}\times \R$ given by
    $$\hat{J}(\mu_{R},z) = (R^{T}\mu_{R}, z).$$

    Let the Hamiltonian $H:T^{*}SO(3)\times \R\rightarrow \R$ be
    $$H(\mu_{R},z) = \frac{1}{2}\tr(\mu_{R}^{T}\mathbb{J}\mu_{R}) + \gamma z.$$

    Let $h:\mathfrak{so}^{*}(3)\times\mathbb{R}  \rightarrow \mathbb{R}$ be given by $h(\mu,z)=\frac{1}{2}\mu^{T}\mathbb{J}\mu + \gamma z$, where we are identifying $\mathfrak{so}^{*}(3)$ with $\R^{3}$ through the map $\Breve{(\cdot)}:\mathbb{R}^{3}\rightarrow \mathfrak{so}^{*}(3)$ satisfying $\langle \Breve{\Pi}, \xi \rangle = \Pi_{i}\xi_{i}$, where $\Pi = (\Pi_{1},\Pi_{2},\Pi_{3})\in \R^{3}$ and $\xi \in \mathfrak{so}(3)$ is identified with $(\xi_{1},\xi_{2},\xi_{3})$ through the hat map given in Example \ref{so3:example}. Then, it is straightforward to check that $h \circ \hat{J} (\mu_{R},z) = H(\mu_{R}, z)$.
\end{example}

\section{Euler-Poincar\'e-Herglotz and Lie-Poisson-Jacobi equations with symmetry breaking}

The purpose of this section is to obtain the reduced Euler-Poincar\'e-Herglotz equations as a generalization of the procedure   followed in \cite{holm1998euler}. Let $X$ be a finite-dimensional vector space. Consider a Lagrangian system evolving on a Lie group $G$ and assume the Lagrangian is not symmetry invariant. At the same time, we consider a representation of $G$ on a dual vector space $X^{*}$.   Coupling to the Lagrangian a parameter depending on vectors in $X^{*}$, we obtain a symmetry invariant Lagrangian function under the action of $G$. Loosely speaking, the Lagrangian system is coupled with vectors in $X^{*}$ that are acted by a Lie group symmetry. The associated action considered at this stage restores the full Lie group symmetry and allows us to apply the semi-direct product reduction theory  \cite{marsden1984semidirect}, \cite{holm1998euler} (see also \cite{gay2010reduction}, \cite{gay2011clebsch}), to obtain the corresponding Euler-Poincar\'e-Herglotz system on the semi-direct product Lie algebra $(\mathfrak{g}\ltimes X^{*})\times\mathbb{R}$. This gives rise to a new reduced system that finds no analogs in classical reduced-order models in contact mechanical systems.

Assume (i) the Lagrangian function $L_{\alpha_0}:TG\times\mathbb{R}\to\mathbb{R}$ is not $G$-invariant but depends on a parameter $\alpha_0$ that may be considered to be an element of the dual space $X^{*}$, where $X$ is a vector space. Hence, we define the extended Lagrangian function as $L_{\textnormal{ext}}: TG\times\mathbb{R}\times X^{*}\to\mathbb{R}$, with $L_{\textnormal{ext}}(\cdot,\alpha_{0}) = L_{\alpha_0}$.
Assume (ii) there is a left representation $\rho$ of $G$ on $X$, i.e., $\rho: G\to\mathrm{GL}(X)$ is a Lie group morphism. Hence, there is a left representation $\rho^{*}$ of $G$ on $X^{*}$, i.e., $\rho^{*}: G\to\mathrm{GL}(X^{*})$, defined using adjoint linear maps. Indeed, for any $g \in G$, we set $\rho^*(g)=\rho({g^{-1}})^{*}\in\mathrm{GL}(X^{*})$, where the right-hand side is the adjoint map of $\rho({g^{-1}})\in\mathrm{GL}(X)$, i.e., for any $x\in X$ and $\alpha \in X^{*}$
\begin{align}
\langle\rho({g^{-1}})^{*}(\alpha),x\rangle = \langle\alpha,\rho({g^{-1}})(x)\rangle.\nonumber
\end{align}
From now, we will denote the image of $g$ by $\rho$ and $\rho^{*}$ by $\rho_{g}$ and $\rho^{*}_{g}$, respectively, and $\rho^{*}$ will be called the adjoint representation.

The adjoint representation induces a left action of $G$ on the extended phase space $TG\times\mathbb{R} \times X^{*}$ defined as
\begin{align}
\Phi: G \times (TG \times\mathbb{R}\times X^{*})&\longrightarrow TG\times\mathbb{R}\times  X^{*},\nonumber\\
(g,(v_h,z,\alpha))&\longmapsto(\hat{\mathcal{L}}_{g}(v_h,z),\rho_{g}^{*}(\alpha)),\label{eq_phi}
\end{align}
where $\hat{\mathcal{L}}_{g}$ has been defined in \eqref{lhat}.
Assume  also that (iii) the extended Lagrangian function is $G$-invariant under \eqref{eq_phi}, i.e., $L_{\textnormal{ext}}\circ\Phi_{g} = L_{\textnormal{ext}}$, for any $g \in G$, or more explicitly {$$L_{\textnormal{ext}}(\hat{\mathcal{L}}_{g}(v_h,z),\rho^{*}_{g}(\alpha)) = L_{\textnormal{ext}}(v_h,z,\alpha),$$} for any $v_h\in TG$ and $\alpha\in X^{*}$. 

We can now obtain the reduced augmented Lagrangian $\ell_{ext}: \mathfrak{g}\times\mathbb{R}\times X^{*}\to\mathbb{R}$, which is given by
\begin{align}
\ell_{ext}(\xi,z,\alpha) := L_{ext}(e,\xi,z,\alpha),
\end{align} 
Critical points of the reduced variational principle for  $\ell_{ext}$ satisfy the following Euler-Poincar{\'e}-Herglotz equations (see Theorem \ref{th5.1} below for a proof)
\begin{align}
\frac{d}{dt}\left(\frac{\delta \ell_{ext}}{\delta \xi}\right) &= \hbox{ad}^{*}_{\xi}\left(\frac{\delta \ell_{ext}}{\partial \xi}\right)+\mathbf{J}_{X}\left(\frac{\delta \ell_{\textnormal{ext}}}{\delta\alpha},\alpha\right)+\frac{\delta \ell_{\textnormal{ext}}}{\delta \xi}\frac{\partial \ell_{\textnormal{ext}}}{\partial z} ,\label{eq_ep}\\
\dot{\alpha} &= {-\rho_{\xi}^{\prime*}(\alpha), \quad \alpha(0) = \alpha_{0}},\label{eq_diff_ep} \\
 \dot{z} & =\ell_{\textnormal{ext}}
\end{align}
where $\mathbf{J}_{X}: T^{*}X\cong X\times X^{*}\to\mathfrak{g}^{*}$ is the momentum map corresponding to the left action of $G$ on $X$ defined using the left representation $\rho$ of $G$ on $X$, i.e., for any $x\in X$, $\xi\in \mathfrak{g}$ and $\alpha\in X^{*}$
\begin{align}
\langle\mathbf{J}_{X}(x,\alpha),\xi\rangle = \langle\alpha,\xi_{X}(x)\rangle,\label{eq_mp}
\end{align}
with $\xi_{X}$ being the infinitesimal generator of the left action of $G$ on $X$ and $\rho^{\prime*}: \mathfrak{g}\to\mathfrak{gl}(X^{*})$ is defined as the adjoint of $\rho^{\prime}$, which is the infinitesimal representation induced by the left representation $\rho$ of $G$ on $X$. Note that the solution to \eqref{eq_diff_ep} is given by {$\alpha(t) = \rho^{*}_{g^{-1}(t)}(\alpha_{0})$}. Also, note that using the notation of \cite{holm1998euler}, \cite{holm2009geometric} we have $\mathbf{J}_{X}(x,\alpha) = x\diamond\alpha$. To summarize, we have the following theorem.

\begin{theorem}\label{th5.1}
	{Let $g:I \rightarrow G$ and $z:I\rightarrow\R$ be smooth curves on an interval $I$. Define the curves $\xi(t)=T_{g(t)}\mathcal{L}_{g^{-1}(t)}(\dot{g}(t))$ and $\alpha(t)=\rho^{*}_{g^{-1}(t)}(\alpha_{0})$, for $\alpha_{0}\in X^{*}$ Let also $g_{0}=g(0)\in G$.} Under assumptions (i)-(iii), the following are equivalent:
	\begin{enumerate}
        \item {The variational principle
		\begin{equation*}
		\delta \int_{0}^{h} L_{\alpha_0}(g(t), \dot{g}(t),z(t)) \ dt=0, \quad \dot{z}=L_{\alpha_0}(g(t), \dot{g}(t),z(t))
		\end{equation*}
		holds for variations of $g$ with fixed endpoints.}
		\item {The curves $g:I \rightarrow G$ and $z:I\rightarrow\R$ satisfy Euler-Poincar{\'e}-Herglotz equations for $L_{\alpha_{0}}$;}
        \item The reduced constrained variational principle
		\begin{equation}\label{reduced:principle2}
		\delta \int_{0}^{h} \ell_{\textnormal{ext}}(\xi(t),z(t),\alpha(t)) \ dt=0, \quad \dot{z}=\ell_{\textnormal{ext}}(\xi(t),z(t), \alpha(t))
		\end{equation}
		holds using variations of $\xi$ and $\alpha$ of the form
            \begin{align}
                \delta{\xi} &= \dot{\eta}+\hbox{ad}_{\xi}\eta,\nonumber\\
                \delta{\alpha} &= {-\rho_\eta^{\prime*}(\alpha)}.\nonumber
            \end{align}
		\item The curves $\xi(t)=T_{g(t)}\mathcal{L}_{g^{-1}(t)}(\dot{g}(t))$ and $z$ satisfy the \textit{Euler-Poincar\'e-Herglotz} equations 
		  \begin{align}
                \frac{d}{dt}\left(\frac{\delta \ell_{ext}}{\delta \xi}\right) &= \hbox{ad}^{*}_{\xi}\left(\frac{\delta \ell_{ext}}{\partial \xi}\right)+\mathbf{J}_{X}\left(\frac{\delta \ell_{\textnormal{ext}}}{\delta\alpha},\alpha\right)+\frac{\delta \ell_{\textnormal{ext}}}{\delta \xi}\frac{\partial \ell_{\textnormal{ext}}}{\partial z},\label{Herglotz:Poincare1}\\
                \dot{\alpha} &= {-\rho_{\xi}^{\prime*}(\alpha)}, \ \ \alpha(0) = \rho^{*}_{g_{0}^{-1}}(\alpha_0), \hspace{5pt} \dot{z}=\ell_{\textnormal{ext}}\label{Herglotz:Poincare2}
            \end{align}
	   \end{enumerate}
\end{theorem}

\begin{proof}

The proof follows arguments similar to the ones given in Theorem \ref{left:trivialized:EL:theorem} and  Theorem \ref{4points}. The equivalence of items $1.$ and $2.$ is general for all contact manifolds, so in particular to $TG\times \R$ (see \cite{anahory2021geometry}).

{We will show that the first variational principle  implies the second constrained variational principle, that is, we prove the equivalence between $1.$ and $3$. Notice that, since $L_{\textnormal{ext}}$ is $G$-invariant under $\Phi_g$, and $\alpha = \rho^{*}_{g^{-1}}(\alpha_{0})$ the integrand in the first variational principle is equal to the integrand in the second constrained variational principle. Indeed, using the defintion of $\ell_{ext}$ and the $G$-invariance, we get
$$\ell_{\textnormal{ext}}(\xi,z,\alpha) = L_{ext}(e,\xi,z,\alpha)=L_{ext}(e,T_{g}\mathcal{L}_{g^{-1}}(\dot{g}),z,\rho^{*}_{g^{-1}}(\alpha_{0}))=L_{\alpha_{0}}(g,\dot{g}, z),$$
where we dropped the time dependence to ease the notation.}

Moreover, all variations of $g$ vanishing at the endpoints induce and are induced by variations of $\xi$ of the form $\delta{\xi} = \dot{\eta}+\ad_{\xi}\eta$, with $\eta(0) = \eta(T) = 0$.

Recall that the infinitesimal variations $\eta$ are given by $\eta = T_{g}\mathcal{L}_{g^{-1}}(\delta{g})$. {This fact is used to write the variations of the curve $\alpha$ as $\delta{\alpha} = -\rho_\eta^{\prime*}(\alpha)$ (see \cite{holm2009geometric}). Hence, the first variational principle implies the second constrained variational principle.}

To conclude, we show the equivalence between $3$. and $4$. Indeed by using the definition of the auxiliar variable $\sigma(t)$, the calculations for Theorem \ref{left:trivialized:EL:theorem}, integration by parts, and the fact that
{\begin{align}
&\left\langle\delta{\alpha},\frac{\delta\ell_{\textnormal{ext}}}{\delta\alpha}\right\rangle = \left\langle -\rho_\eta^{\prime*}(\alpha),\frac{\delta \ell_{\textnormal{ext}}}{\delta\alpha}\right\rangle = \left\langle \alpha,\rho_\eta^{\prime}\left(\frac{\delta \ell_{\textnormal{ext}}}{\delta\alpha}\right) \right\rangle
= \Big\langle\mathbf{J}_{X}\left(\frac{\delta \ell_{\textnormal{ext}}}{\delta\alpha},\alpha\right),\eta\Big\rangle,\nonumber
\end{align}
where we used the infinitesimal version of the definition of the adjoint representation,i.e.,
$$\left\langle \rho_\eta^{\prime*}(\alpha),x\right\rangle = -\left\langle \alpha,\rho_\eta^{\prime}\left(x\right) \right\rangle, \quad \forall x\in X, \alpha\in X^{*}, \eta\in \mathfrak{g},$$}
we deduce that the reduced Herglotz variational principle holds if and only if
\begin{equation}
\sigma(t)\left[\text{ad}_{\xi}^{*}\frac{\delta l}{\partial \xi}-\frac{d}{dt}\left(\frac{\delta l}{\delta \xi}\right)+\mathbf{J}_{X}\left(\frac{\delta \ell_{\textnormal{ext}}}{\delta\alpha},\alpha\right)+\frac{\partial l}{\partial z}\frac{\delta l}{\delta \xi}\right] = 0.
\end{equation}  By taking the time derivative of $\alpha$, {we get
$\dot{\alpha} = -\rho_{\xi}^{\prime*}(\alpha), \hspace{5pt} \alpha(0) = \rho^{*}_{g_{0}^{-1}}(\alpha_0)$} and, since the auxiliary function $\sigma$ is nowhere zero, the result follows.
\end{proof}

There is a special cases of Theorem \ref{th5.1} that is of particular interest and we state it as corollary.

\begin{corollary}\label{cor5}
Let $X = \mathfrak{g}$ and $\rho$ be the adjoint representation of $G$ on $X$, i.e., $\rho_{g} = \hbox{Ad}_{g}$, for any $g\in G$. Then, the Euler--Poincar{\'e}-Herglotz equations \eqref{Herglotz:Poincare1}--\eqref{Herglotz:Poincare2} give the following equations
\begin{align}
\frac{d}{dt}\left(\frac{\delta\ell_{\textnormal{ext}}}{\delta \xi}\right) &= \ad^{*}_{\xi}\left(\frac{\delta \ell_{\textnormal{ext}}}{\delta \xi}\right)-\ad_{\frac{\delta \ell_{\textnormal{ext}}}{\delta\alpha}}^{*}\alpha+\frac{\delta \ell_{\textnormal{ext}}}{\delta \xi}\frac{\partial \ell_{\textnormal{ext}}}{\partial z},\nonumber\\
\dot{\alpha} &= {-\ad_{\xi}^{*}\alpha, \hspace{5pt} \alpha(0) = \Ad_{g_{0}^{-1}}^{*}\alpha_{0}},\,\dot{z}=\ell_{\textnormal{ext}}.\nonumber
\end{align}
\end{corollary}

\begin{proof}
We first begin by noting that $\alpha\in X^{*} = \mathfrak{g}^{*}$ and $\rho^{*}$ is the coadjoint representation of $G$ on $X^{*}$, i.e., $\rho^{*}_{g} = \Ad_{g}^{*}$, for any $g\in G$. We also have $\rho_{\xi}^{\prime} = \ad_{\xi}$, for any $\xi\in\mathfrak{g}$ and it follows that $\xi_{X}(x) = \ad_{\xi}x$, for any $x\in X$. From \eqref{eq_mp}, we have 
\begin{align}
\langle\mathbf{J}_{X}(x,\alpha),\xi\rangle &= \langle\alpha,\ad_{\xi}x\rangle\nonumber\\
&= \langle\alpha,-\ad_{x}\xi\rangle\nonumber\\
&= \langle-\ad_{x}^{*}\alpha,\xi\rangle,\nonumber
\end{align}
which gives $\mathbf{J}_{X}(x,\alpha) = -\ad_{x}^{*}\alpha$.
\end{proof}


\subsection{Reduced Legendre transformation}
If we assume that the reduced Lagrangian $\ell_{\textnormal{ext}}$ is hyper-regular, then we can obtain the reduced Hamiltonian $h_{\textnormal{ext}}: \mathfrak{g}^{*}\times\mathbb{R}\times X^{*}\to\mathbb{R}$ (by using the reduced Legendre transformation) given by
\begin{align}
h_{\textnormal{ext}}(\mu,z,\alpha) = \langle\mu,\xi\rangle-\ell_{\textnormal{ext}}(\xi,z,\alpha),\nonumber
\end{align}
where $\mu = \frac{\partial\ell_{\textnormal{ext}}}{\partial \xi}$, with $\mu(\cdot)\in C^{1}([0,T],\mathfrak{g}^{*})$. The Euler--Poincar{\'e}-Herglotz equations \eqref{Herglotz:Poincare1}--\eqref{Herglotz:Poincare2} can now be written as the Lie--Poisson--Jacobi equations, which are given by
\begin{align}
\dot{\mu} & =\emph{ad}^{*}_{\frac{\delta h_{\textnormal{ext}}}{\delta \mu}}\mu-\mu\frac{\partial h_{\textnormal{ext}}}{\partial z}-\mathbf{J}_{X}\left(\frac{\delta h_{\textnormal{ext}}}{\delta\alpha},\alpha\right),\\	\dot{z} &= \langle \mu,\frac{\delta h_{\textnormal{ext}}}{\delta \mu} \rangle-h_{\textnormal{ext}}(\mu,z,\alpha).\label{eq_lp}\\
\dot{\alpha} &= {-\rho^{\prime*}_{\frac{\delta h_{\textnormal{ext}}}{\delta\alpha}}(\alpha), \hspace{5pt} \alpha(0) = \rho^{*}_{g_{0}^{-1}}(\alpha_{0})}.\label{eq_diff_lp}
\end{align}


\subsection{Lie-Poisson-Jacobi equations for systems with symmetry breaking}

We can define a Jacobi bracket on the space of smooth functions on $\mathfrak{g}^{*}\times\mathbb{R}\times X^{*}$ with respect to which the previous equations are just the Hamiltonian equations. Consider the bracket
\begin{equation*}
    \begin{split}
        \{f,g\}_{X^{*}} (\mu,z,\alpha)=\left\langle \mu,\left[ \frac{\delta f}{\delta \mu},\frac{\delta g}{\delta \mu} \right] \right\rangle & + \left\langle \mu,\frac{\delta f}{\delta \mu} \right\rangle \frac{\partial g}{\partial z} - \left\langle \mu,\frac{\delta g}{\delta \mu} \right\rangle \frac{\partial f}{\partial z} \\
        & + \left\langle \alpha, \frac{\delta f}{\delta \mu}\frac{\delta g}{\delta \alpha} - \frac{\delta f}{\delta \alpha}\frac{\delta g}{\delta \mu}\right\rangle
    \end{split}
\end{equation*}
This bracket falls into the same definition of Lie-Poisson-Jacobi bracket as the one given in Definition \ref{LPJ-bracket-def}. In this case, the associated Lie algebra is $\mathfrak{g}\ltimes X$ (see \cite{holm2009geometric} for more details on the subject of semi-direct products).

\begin{proposition}
    The Hamiltonian equations with respect to the Lie-Poisson-Jacobi bracket $\{\cdot,\cdot\}_{X^{*}}$ and the Hamiltonian function $h$ on $\mathfrak{g}^{*}\times\mathbb{R}\times X^{*}$ are
    \begin{align*}
        \dot{\mu} & =\emph{ad}^{*}_{\frac{\delta h}{\delta \mu}}\mu-\mu\frac{\partial h}{\partial z}-\mathbf{J}_{X}\left(\frac{\delta h}{\delta\alpha},\alpha\right),\\	\dot{z} &= \langle \mu,\frac{\delta h}{\delta \mu} \rangle-h(\mu,z,\alpha).\label{eq_lp}\\
        \dot{\alpha} &= {-\rho^{\prime*}_{\frac{\delta h}{\delta \mu}}(\alpha)}.
\end{align*}
\end{proposition}

\begin{proof}
Let $f\in C^{\infty}(\mathfrak{g}^{*}\times\mathbb{R}\times X^{*})$. Then,
\begin{equation*}
		\dot{f}= \langle df(\mu,z, \alpha), (\dot{\mu},\dot{z},\dot{\alpha}) \rangle = \left\langle \dot{\mu},\frac{\delta f}{\delta \mu} \right\rangle + \frac{\partial f}{\partial z}\dot{z} + \left\langle \dot{\alpha},\frac{\delta f}{\delta \alpha} \right\rangle.
	\end{equation*}
	Also, using the definition of the Jacobi bracket we have that
	\begin{equation*}
        \begin{split}
            \{h,f\}_{X^{*}}(\mu, z, \alpha)=\left\langle \mu,\left[ \frac{\delta h}{\delta \mu},\frac{\delta f}{\delta \mu} \right] \right\rangle & + \left( \left\langle \mu,\frac{\delta h}{\delta \mu} \right\rangle-h \right) \frac{\partial f}{\partial z} - \left( \left\langle \mu,\frac{\delta f}{\delta \mu} \right\rangle - f \right) \frac{\partial h}{\partial z} \\
            & + \left\langle \alpha, \frac{\delta h}{\delta \mu}\frac{\delta f}{\delta \alpha} - \frac{\delta h}{\delta \alpha}\frac{\delta f}{\delta \mu}\right\rangle.
        \end{split}
	\end{equation*}
	Thus using \eqref{hamiltonian:evolution} we have that
	\begin{equation*}
        \begin{split}
            \dot{f}= \{h,f\}_{X^{*}}- f\frac{\partial h}{\partial z} & =\left\langle \text{ad}_{\frac{\delta h}{\delta \mu}}^{*}\mu-\mu\frac{\partial h}{\partial z},\frac{\delta f}{\delta \mu} \right\rangle + \left( \left\langle \mu,\frac{\delta h}{\delta \mu} \right\rangle-h \right) \frac{\partial f}{\partial z} \\
            & + \left\langle \alpha, \rho_{\frac{\delta h}{\delta \mu}}^{'}\frac{\delta f}{\delta \alpha}\right\rangle - \left\langle \alpha,\rho_{\frac{\delta f}{\delta \mu}}^{'}\frac{\delta h}{\delta \alpha}\right\rangle \\
            & = \left\langle \text{ad}_{\frac{\delta h}{\delta \mu}}^{*}\mu-\mu\frac{\partial h}{\partial z},\frac{\delta f}{\delta \mu} \right\rangle + \left( \left\langle \mu,\frac{\delta h}{\delta \mu} \right\rangle-h \right) \frac{\partial f}{\partial z} \\
            & + \left\langle \rho_{\frac{\delta h}{\delta \mu}}^{'*}\alpha, \frac{\delta f}{\delta \alpha}\right\rangle - \left\langle\mathbf{J}_{X}\left(\frac{\delta h}{\delta \alpha},\alpha \right), \frac{\delta f}{\delta \mu}\right\rangle \\ 
        \end{split}
	\end{equation*}
	Since the function $f$ is arbitrary, the result follows.
\end{proof}

\subsection{The Heavy Top with dissipation}

Consider $G = \mathrm{SO}(3)$, with Lagrangian function given by 
\begin{equation}
{L_{\mathbf{e}_{3}}(g,\xi,z) = \frac{1}{2}\langle \xi,\mathbb{I}\xi\rangle - m\mathbf{g}l\langle\mathbf{e}_{3},R\boldsymbol{\chi}\rangle-\gamma z},\nonumber
\end{equation}
where $R\in SO(3)$, $\langle\xi,\xi\rangle=\tr(\xi^{T}\xi)$, for any $\xi\in\mathfrak{g} = \mathfrak{so}(3)$, $\xi(\cdot) = \sum_{i=1}^{3}\xi^{i}(\cdot)e_{i}$, with the elements of the basis of $\mathfrak{g}$  given by
\begin{align}
e_{1} = \begin{bmatrix}
        0 & \phantom{-}0 & \phantom{-}0 \\
        0 & \phantom{-}0 & -1 \\
        0 & \phantom{-}1 & \phantom{-}0
        \end{bmatrix}, \hspace{5pt} 
e_{2} = \begin{bmatrix}
        \phantom{-}0 & \phantom{-}0 & \phantom{-}1 \\
        \phantom{-}0 & \phantom{-}0 & \phantom{-}0 \\
        -1 & \phantom{-}0 & \phantom{-}0
      \end{bmatrix}, \hspace{5pt}
e_{3} = \begin{bmatrix}
        0 & -1 & \phantom{-}0 \\
        1 & \phantom{-}0 & \phantom{-}0 \\
        0 & \phantom{-}0 & \phantom{-}0
        \end{bmatrix},\nonumber
\end{align}
which satisfy
\begin{align}
[e_{1},e_{2}] = e_{3}, \hspace{5pt} [e_{2},e_{3}] = e_{1}, \hspace{5pt} [e_{3},e_{1}] = e_{2},\nonumber
\end{align}
$\mathbb{I}: \mathfrak{g}\to\mathfrak{g}^{*} = \mathfrak{so}(3)^{*}$ is the inertia tensor of the top (it is calculated with respect to the pivot, which is not, in general, the center of mass), $m$ is the mass of the body, $\mathbf{g}$ is the acceleration due to gravity, $\mathbf{e}_{3}$ is the vertical unit vector, $\boldsymbol{\chi}\in\mathbb{R}^{3}$ is the unit vector from the point of support to the direction of the body's center of mass (constant) in body coordinates, $l$ is the length of the line segment between these two points and $\gamma\in\mathbb{R}$.

Under the dual pairing, where $\langle\alpha,\xi\rangle=\tr(\alpha\xi)$, for any $\xi\in\mathfrak{g}$ and $\alpha\in\mathfrak{g}^{*}$, the elements of the basis of $\mathfrak{g}^{*}$ are given by
\begin{align}
e^{1} = \begin{bmatrix}
        0 & \phantom{-}0 & \phantom{-}0\\[4pt]
        0 & \phantom{-}0 & \phantom{-}\dfrac{1}{2}\\[4pt]
        0 & -\dfrac{1}{2} & \phantom{-}0
        \end{bmatrix}, \hspace{5pt} 
e^{2} = \begin{bmatrix}
        \phantom{-}0 & \phantom{-}0 & -\dfrac{1}{2}\\[8pt]
        \phantom{-}0 & \phantom{-}0 & \phantom{-}0\\[4pt]
        \phantom{-}\dfrac{1}{2} & \phantom{-}0 & \phantom{-}0
        \end{bmatrix}, \hspace{5pt}
e^{3} = \begin{bmatrix}
        \phantom{-}0 & \phantom{-}\dfrac{1}{2} & \phantom{-}0\\[4pt]
        -\dfrac{1}{2} & \phantom{-}0 & \phantom{-}0\\[8pt]
        \phantom{-}0 & \phantom{-}0 & \phantom{-}0
        \end{bmatrix}.\nonumber
\end{align}

{It is easy to verify that the Lagrangian function is invariant under the left action of the isotropy group
\begin{align}
G_{\alpha_{0}}=\{g\in G\mid\Ad_{g}^{*}\alpha_{0} = \alpha_{0}\}\cong\mathrm{SO}(2),\nonumber
\end{align}
i.e., rotations about the vertical axis $\mathbf{e}_{3}$, but not $G$-invariant. The potential function breaks the symmetry partially. However, the framework we introduced before will help us to restore the full symmetry.}

Let $X = \mathfrak{g}$ and $\rho$ be the adjoint representation of $G$ on $X$, i.e., $\rho_{g} = \Ad_{g}$, for any $g\in G$. So, $\rho^{*}$ is the coadjoint representation of $G$ on $X^{*} = \mathfrak{g}^{*}$, i.e., $\rho^{*}_{g} = \Ad_{g}^{*}$, for any $g\in G$. Let the extended Lagrangian function $L_{\textnormal{ext}}: G\times\mathfrak{g}\times X^{*}\to\mathbb{R}$ be given as follows
\begin{align}
{L_{\textnormal{ext}}(g,\xi,\breve{\alpha}) = \frac{1}{2}\langle \xi,\mathbb{I}\xi\rangle - m\mathbf{g}l\langle R^{-1}\alpha,\boldsymbol{\chi}\rangle-\gamma z},\nonumber
\end{align}
where $\alpha\in\mathbb{R}^{3}$ is identified with $\breve{\alpha}\in\mathfrak{g}^{*}$. If we set $\breve{\alpha}_{0} = -2e^{3}$ in the extended Lagrangian function, we recover our original Lagrangian function. It is easy to verify that the extended Lagrangian function is $G$-invariant under \eqref{eq_phi}, i.e., $L_{\textnormal{ext}}\circ\Phi_{g} = L_{\textnormal{ext}}$, for any $g \in G$.

Thus, we can define the reduced extended Lagrangian function to be
$$\ell_{ext}(\xi, \alpha, z) = \frac{1}{2}\langle \xi,\mathbb{I}\xi\rangle - m\mathbf{g}l\langle \alpha,\boldsymbol{\chi}\rangle-\gamma z,$$
and by Corollary \ref{cor5}, the Euler--Poincar{\'e}-Herglotz equations (under the identifications $\mathfrak{g}\cong\mathbb{R}^{3}$ and $\mathfrak{g}^{*}\cong\mathbb{R}^{3}$) give the following equations
\begin{align}
\mathbb{I}\dot{\xi} &= {\ad_{\xi}^{*}\mathbb{I} \xi-\ad_{\frac{\partial L_{\textnormal{ext}}}{\partial\alpha}}^{*}\alpha- \gamma \mathbb{I} \xi},\label{eq_ep1_ht}\\
\dot{\alpha} &= {-\ad_{\xi}^{*}\alpha, \hspace{5pt} \alpha(0) = \Ad_{g_{0}^{-1}}^{*}\alpha_{0}},\quad \dot{z}=\ell_{\textnormal{ext}}.\label{eq_diff_ht}
\end{align}

Using the expressions (for more details, see \cite{holm2009geometric}, \cite{marsden2013introduction})
\begin{align}
\ad_{\xi}^{*}\mathbb{I}\xi = \mathbb{I}\xi\times \xi, \hspace{5pt} \ad_{\frac{\partial L_{\textnormal{ext}}}{\partial\alpha}}^{*}\alpha = -m\mathbf{g}l\hspace{2pt}\alpha\times\boldsymbol{\chi}, \hspace{5pt} \ad_{\xi}^{*}\alpha = \alpha\times \xi, \hspace{5pt} {\Ad_{g_{0}^{-1}}^{*}\alpha_{0} = g_{0}\alpha_{0}}\nonumber
\end{align}
and $\mathbb{I}\in\mathbb{R}^{3\times 3}$ is the inertia matrix,  \eqref{eq_ep1_ht}--\eqref{eq_diff_ht} give the following equations
\begin{align}
\mathbb{I}\dot{\xi} &= \mathbb{I}\xi\times \xi-m\mathbf{g}l\hspace{2pt}\boldsymbol{\chi}\times\alpha - \gamma \mathbb{I} \xi,\nonumber\\
\dot{\alpha} &={ -\alpha\times \xi, \hspace{5pt} \alpha(0) = g_{0} \alpha_{0}},\,\,\dot{z}=\ell_{\textnormal{ext}}.\nonumber
\end{align}

In the case $C=\gamma\mathbb{I}$ we obtain the equations for rigid body attitude dynamics studied in \cite{shen2003asymptotic}, \cite{cho2003mathematical} for the triaxial attitude control testbed given in \cite{bernstein2001development}.

We do not give all the details and leave it up to reader to verify that the Lie--Poisson--Jacobi equations \eqref{eq_lp}--\eqref{eq_diff_lp} are given by 
\begin{align}
\dot{\mu} &= \mu\times\mathbb{I}^{-1}\mu-m\mathbf{g}l\hspace{2pt}\boldsymbol{\chi}\times\alpha-\gamma\mu ,\nonumber\\
\dot{\alpha} &= {-\alpha\times\mathbb{I}^{-1}\mu, \hspace{5pt} \alpha(0) = g_{0}\alpha_{0}},\,\,\quad \dot{z} = \frac{1}{2}\mu^{T}\mathbb{I}\mu -m\mathbf{g}l\langle g^{-1}\alpha,\boldsymbol{\chi}\rangle- \gamma z.\nonumber
\end{align}

\section{Conclusions and Future Work}

In this paper, we have established a geometric framework to study reduced contact dynamics on Lie groups. In future work we wish to employ this framework to study \begin{itemize}
   \item the stability of reduced contact systems,
    \item Noether's theorems and dissipated quantities - the analogous notion of conserved quantities in symplectic dynamics.
    \item contact structure-preserving discretization schemes of contact dynamics on Lie groups and semi-direct products. Discretization schemes complying with the reduction process could also be used to find and numerically implement optimal control policies in systems with dissipation.
\end{itemize}




\section*{Acknowledgments}  The authors acknowledge financial support from Grant PID2019-106715GB-C21 funded by MCIN/AEI/ 10.13039/501100011033. J.C. Marrero and E. Padr\'on acknowledge financial support from the Spanish Ministry of Science and Innovation under grant PGC2018-098265-B-C32.

\printbibliography

\end{document}